\documentclass[11pt]{amsart}
\usepackage{amsmath,amssymb,amsfonts,amsthm,amscd,textcomp,times,booktabs}
\usepackage[dvipdfmx]{graphicx}
\usepackage{braket}
\usepackage{enumerate}
\usepackage{color,float}
\usepackage{bm}
\usepackage{multirow}
\usepackage{arydshln}
\usepackage[all]{xy}
\usepackage[normalem]{ulem}
\usepackage{mathtools}
\newtheorem{theorem}{Theorem}

\newtheorem{proposition}[theorem]{Proposition}
\newtheorem{lemma}[theorem]{Lemma}

\theoremstyle{definition}
\newtheorem{definition}[theorem]{Definition}
\newtheorem{remark}[theorem]{Remark}

\newtheorem{construction}[theorem]{Construction}
\newtheorem{convention}[theorem]{Convention}

\numberwithin{equation}{section}
\numberwithin{theorem}{section}

\usepackage{scalerel}
\usepackage{stackengine}
\stackMath
\newsavebox\tmpbox
\newcommand\arc[1]{\ThisStyle{\sbox\tmpbox{$\SavedStyle#1$}\stackon[0pt]{\usebox{\tmpbox}}{\stretchto{\scaleto{
\scalerel*[\wd\tmpbox]{\mkern-.8mu\frown\mkern-.8mu}{\rule[-\textheight/2]{1ex}{\textheight}}}{\textheight}}{0.8ex}}}}
\newcommand{\Lt}{\ensuremath{\mathrm{L}}}
\newcommand{\Rt}{\ensuremath{\mathrm{R}}}
\newcommand{\norm}[1]{\ensuremath{\left| #1 \right|}}
\newcommand{\ora}[1]{\ensuremath{\overrightarrow{#1}}}

\newcommand{\proj}{\ensuremath{\mathrm{proj}}}
\newcommand{\can}{\ensuremath{\mathrm{c}}}

\begin{document}
\title[Geometric onstruction of canonical $3$D gadgets in origami extrusions]
{Geometric construction of canonical $3$D gadgets\\ in origami extrusions}
\author{Mamoru Doi}
\address{11-9-302 Yumoto-cho, Takarazuka, Hyogo 665-0003, Japan}
\email{doi.mamoru@gmail.com}
\maketitle
\noindent{\bfseries Abstract.}~
In a series of our three previous papers, we presented several constructions of positive and negative $3$D gadgets in origami extrusions
which create with two simple outgoing pleats a top face parallel to the ambient paper and two side faces sharing a ridge,
where a $3$D gadget is said to be positive (resp. negative)
if the top face of the resulting gadget seen from the front side lies above (resp. below) the ambient paper.
For any possible set of angle parameters, we obtained an infinite number of positive $3$D gadgets in our second paper,
while we obtained a unique negative $3$D gadget by our third construction in our third paper.

In this paper we present a geometric (ruler and compass) construction of our third negative $3$D gadgets, 
while the construction presented in our third paper was a numerical one using a rather complicated formula.
Also, we prove that there exists a unique positive $3$D gadget corresponding to each of our third negative ones.
Thus we obtain a canonical pair of a positive and a negative $3$D gadget.
The proof is based on a geometric redefinition of the critical angles which we introduced in constructing our positive $3$D gadgets.
This redefinition also enables us to give a simplified proof of the existence theorem of our positive $3$D gadgets in our second paper.
As an application, we can construct a positive and a negative extrusion from a common crease pattern by using the canonical counterparts,
as long as there arise no interferences.
\section{Introduction}
This is the fourth in a series of our papers on constructions of $3$D gadgets in origami extrusions,
where an origami extrusion is a folding of a $3$D object in the middle of a flat piece of paper,
and $3$D gadgets are ingredients for origami extrusions which create faces with solid angles.
In our three previous papers \cite{Doi19}, \cite{Doi20} and \cite{Doi21}, we studied $3$D gadgets which create with two simple outgoing pleats
a top face parallel to the ambient paper and two side faces sharing a ridge,
where a simple pleat consists of a mountain and a valley fold which are parallel to each other.
If the top face of a $3$D gadget seen from the front side lies above (resp. below) the ambient paper,
then the $3$D gadget is said to be \emph{positive} (resp. \emph{negative}).

In our second paper \cite{Doi20}, we presented a construction of positive $3$D gadgets with flat back sides
extending the construction in \cite{Doi19}, which improve the $3$D gadgets developed by Natan with a supporting pyramid on the back side \cite{Natan},
where we named our positive $3$D gadgets (positive) \emph{origon gadgets} or simply (positive) \emph{origons}.
Meanwhile, in our third paper \cite{Doi21}, we presented three constructions of negative $3$D gadgets in origami extrusions
in addition to two known constructions before ours, of which we also extended Cheng's construction \cite{Cheng}.
Let us also call negative $3$D gadgets by our three constructions \emph{negative origon gadgets} or simply \emph{negative origons}.

The constructions of the crease patterns of positive and our third negative origon gadgets are characterized by the `dividing point' of the circular arc
which has a radius of the same length as the ridge and subtains the angle formed between the side faces in the development.
In our second paper \cite{Doi20}, we introduced the critical angles by which the possible range of the dividing point is determined.
Then we proved that for any possible set of angle parameters,
there exist infinitely many positive origons which are distinguished by the choice of the dividing point.
Meanwhile, in our third paper we proved that for any possible set of angle parameters, our third construction provides a unique negative origon.
Let us call a negative origon gadget uniquely obtained in our third construction a \emph{canonical} negative origon (gadget).
The unique dividing point for a canonical negative origon is specified by the rotation angle from either endpoint of the arc.
However, the angle was not given geometrically but numerically by a rather complicated formula.

In this paper we present a geometric (ruler and compass) construction of canonical negative origon gadgets in Construction $\ref{const:neg_can}$.
This solves the problem which came up in the conclusion of \cite{Doi21}.
Also, we prove in Theorem $\ref{thm:exist_pos_can}$ that for any canonical negative origon gadget,
there exists a positive one with the same set of angle parameters and the same choice of the dividing point as the negative one.
This result yields a \emph{canonical} positive origon among infinitely many compatible ones,
and a \emph{canonical pair} consisting of a positive and a negative origon which are both canonical and correspond to each other.
The proof is done by showing that the dividing point lies in the range of constructibility of positive origons,
which is determined by the critical angles.
For this purpose, we give in Definition $\ref{def:crit_geom}$ a geometric redefinition of the critical angles,
which directly determines the range of constructibility.
Using this redefinition, which is simpler than the original geometric definition and its numerical rephrasing given in \cite{Doi20},
we can give a simplified proof of the existence theorem of positive origons in Theorem $\ref{thm:zeta_L+R}$.

As an application of canonical pairs, we can construct a positive and a negative extrusion made of origon gadgets
from a common crease pattern by replacing each origon with its canonical counterpart.

\renewcommand{\theenumi}{\arabic{enumi}}
\begin{convention}\label{conv}
We will use the following conventions.
\begin{enumerate}
\item As in our previous papers, we will use subscript $\sigma$ for `L' and `R' standing for `left' and `right' respectively,
and $\sigma'$ for the other side of $\sigma$, that is,
\begin{equation*}
\sigma'=\begin{dcases}\Rt&\text{if }\sigma =\Lt ,\\
\Lt&\text{if }\sigma =\Rt .\end{dcases}
\end{equation*}
\item We will use `development' for the flat piece of paper obtained by developing a $3$D gadget, which includes not only the net of the extruded object
but also the creases hidden behind after the folding.
\item We will determine the standard side from which a $3$D gadget is seen so that the extruded faces lie above the ambient paper.
Thus we will usually consider the development of a positive (resp. negative) gadget seen from the front (resp. back) side.
\item If we denote a circular arc with endpoints $A$ and $B$ by $\arc{AB}$,
then we will regard the arc as swept \emph{counterclockwise} from $A$ to $B$.
\item If we denote an angle by $\angle_+BAC$ (resp. $\angle_-BAC$),
then we will regard the angle as the rotation angle around $A$ of the ray starting from $A$ through $B$ to the ray starting from $A$ through $C$
measured counterclockwise (resp. clockwise), which may take a negative value depending on the range of the angle under consideration.
\end{enumerate}
\end{convention}

To fix the conditions and the construction of typical points, lines and angles common to all of the later constructions of origon gadgets,
we prepare the following.
\begin{construction}\label{const:cond}
We consider a net on a flat piece of paper as in Figure $\ref{fig:dev_0}$, where $\alpha ,\beta_\Lt ,\beta_\Rt ,\gamma\in (0,\pi )$.
We require the following conditions on $\alpha$ (or $\gamma$), $\beta_\Lt ,\beta_\Rt ,\delta_\Lt$ and $\delta_\Rt$.
\renewcommand{\theenumi}{\roman{enumi}}
\begin{enumerate}
\item $\alpha <\beta_\Lt + \beta_\Rt$, $\beta_\Lt <\alpha +\beta_\Rt$ and $\beta_\Rt <\alpha+ \beta_\Lt$, or equivalently,
$\beta_\Lt +\gamma /2<\pi$, $\beta_\Rt +\gamma /2<\pi$ and $\beta_\Lt +\beta_\Rt +\gamma /2>\pi$.
\item $\alpha +\beta_\Lt +\beta_\Rt <2\pi$, or equivalently, $\gamma >0$.
\end{enumerate}
To construct a negative gadget from the above net, we prescribe its simple outgoing pleats by introducing parameters $\delta_\sigma$
for their changes from the direction of $\ora{AB_\sigma}$ for $\sigma =\Lt ,\Rt$, for which we further require the following conditions.
\begin{enumerate}
\item[(iii.a)] $\delta_\Lt\geqslant 0$ and $\delta_\Rt\geqslant 0$,
where we take clockwise (resp. counterclockwise) direction as positive for $\sigma =\Lt$ (resp. $\sigma=\Rt$).
\item[(iii.b)] $\delta_\Lt <\beta_\Lt$ and $\delta_\Lt <\beta_\Lt$.
\item[(iii.c)] $\alpha +\beta_\Lt +\beta_\Rt -\delta_\Lt -\delta_\Rt >\pi$, or equivalently, $\gamma +\delta_\Lt +\delta_\Rt <\pi$.
\end{enumerate}
In particular, if $\delta_\Lt =\delta_\Rt =0$, then conditions (iii.a)--(iii.c) are simplified as
\begin{enumerate}
\setcounter{enumi}{2}
\item $\alpha +\beta_\Lt +\beta_\Rt >\pi$, or equivalently, $\gamma <\pi$.
\end{enumerate}
Then we construct the creases of simple pleats $(\ell_\Lt ,m_\Lt )$ and $(\ell_\Rt ,m_\Rt )$ consisting of a mountain and a valley fold parallel to each other,
which we prescribe as the outgoing pleats of positive and negative origons as follows, where we regard $\sigma$ as taking both $\Lt$ and $\Rt$.
\renewcommand{\theenumi}{\arabic{enumi}}
\begin{enumerate}
\item Draw a ray $\ell_\Lt$ starting from $B_\Lt$ and going to the direction of $\ora{AB_\Lt}$ followed by a clockwise rotation by $\delta_\Lt$.
Also, draw a ray $\ell_\Rt$ starting from $B_\Rt$ and going to the direction of $\ora{AB_\Rt}$ followed by a counterclockwise rotation by $\delta_\Rt$.
Following Convention $\ref{conv}$, $(5)$, we can also define $\ell_\Lt$ (resp. $\ell_\Rt$) as a ray starting from $B_\Lt$ (resp. $B_\Rt$) such that
$\angle_+AB_\Lt\ell_\Lt =\pi -\delta_\Lt$ (resp. $\angle_-AB_\Rt\ell_\Rt =\pi -\delta_\Rt$).
\item Draw a perpendicular to $\ell_\sigma$ through $B_\sigma$ for both $\sigma =\Lt ,\Rt$, letting $C$ be the intersection point.
\item Draw a perpendicular bisector $m_\sigma$ to segment $B_\sigma C$ for both $\sigma =\Lt ,\Rt$, letting $P$ be the intersection point.
Since $P$ is the excenter of $\triangle C B_\Lt B_\Rt$,
segment $AP$ is a perpendicular bisector of segment $B_\Lt B_\Rt$ and also a bisector of $\angle B_\Lt AB_\Rt$.
\end{enumerate}
The resulting creases are shown as solid lines in Figure $\ref{fig:dev_1}$.
\end{construction}
\addtocounter{theorem}{1}
\begin{figure}[htbp]
\centering\includegraphics[width=0.75\hsize]{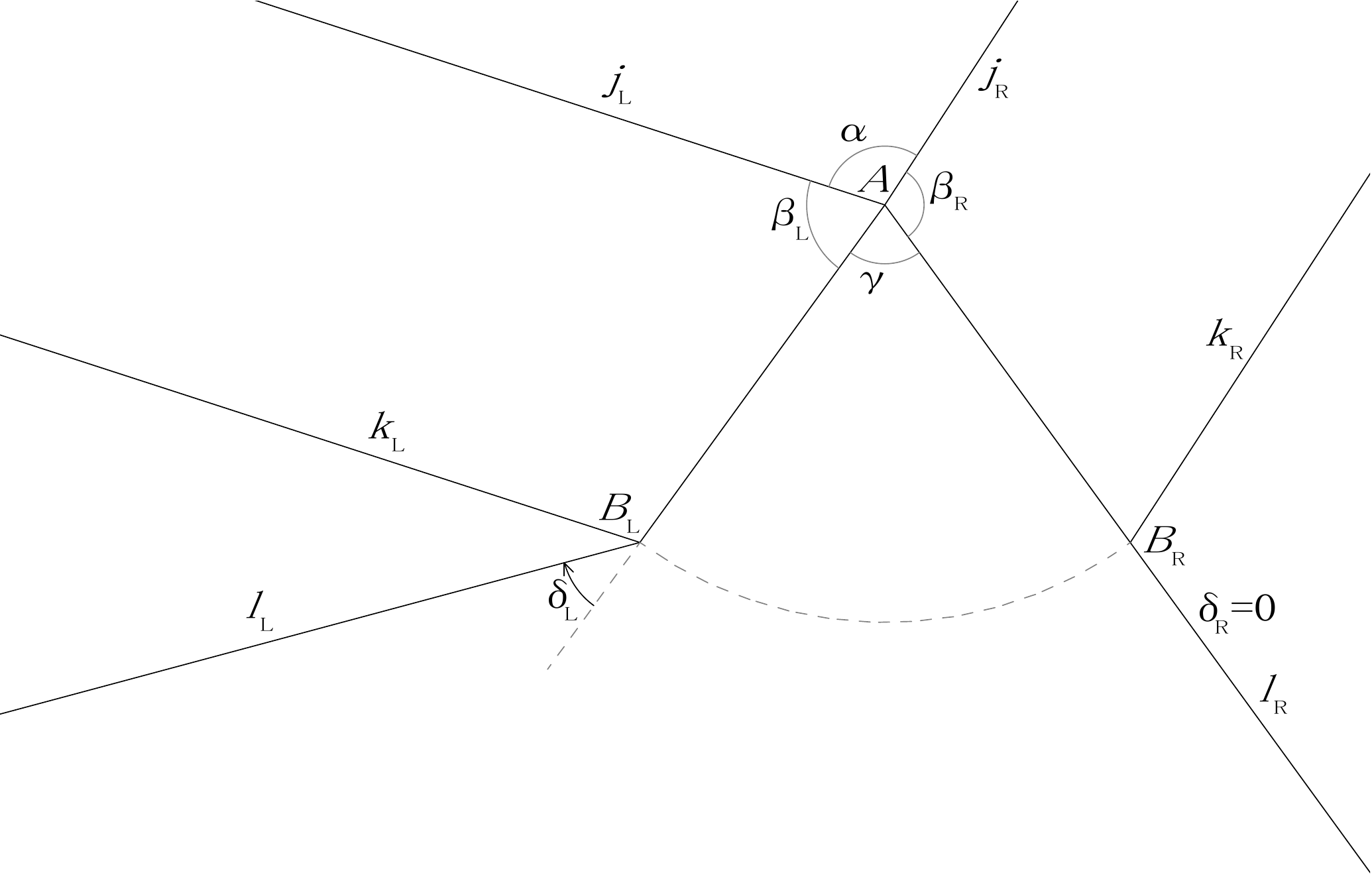}
\caption{Net of a $3$D object which we want to extrude with an origon gadget}
\label{fig:dev_0}
\end{figure}
\addtocounter{theorem}{1}
\begin{figure}[htbp]
\centering\includegraphics[width=0.75\hsize]{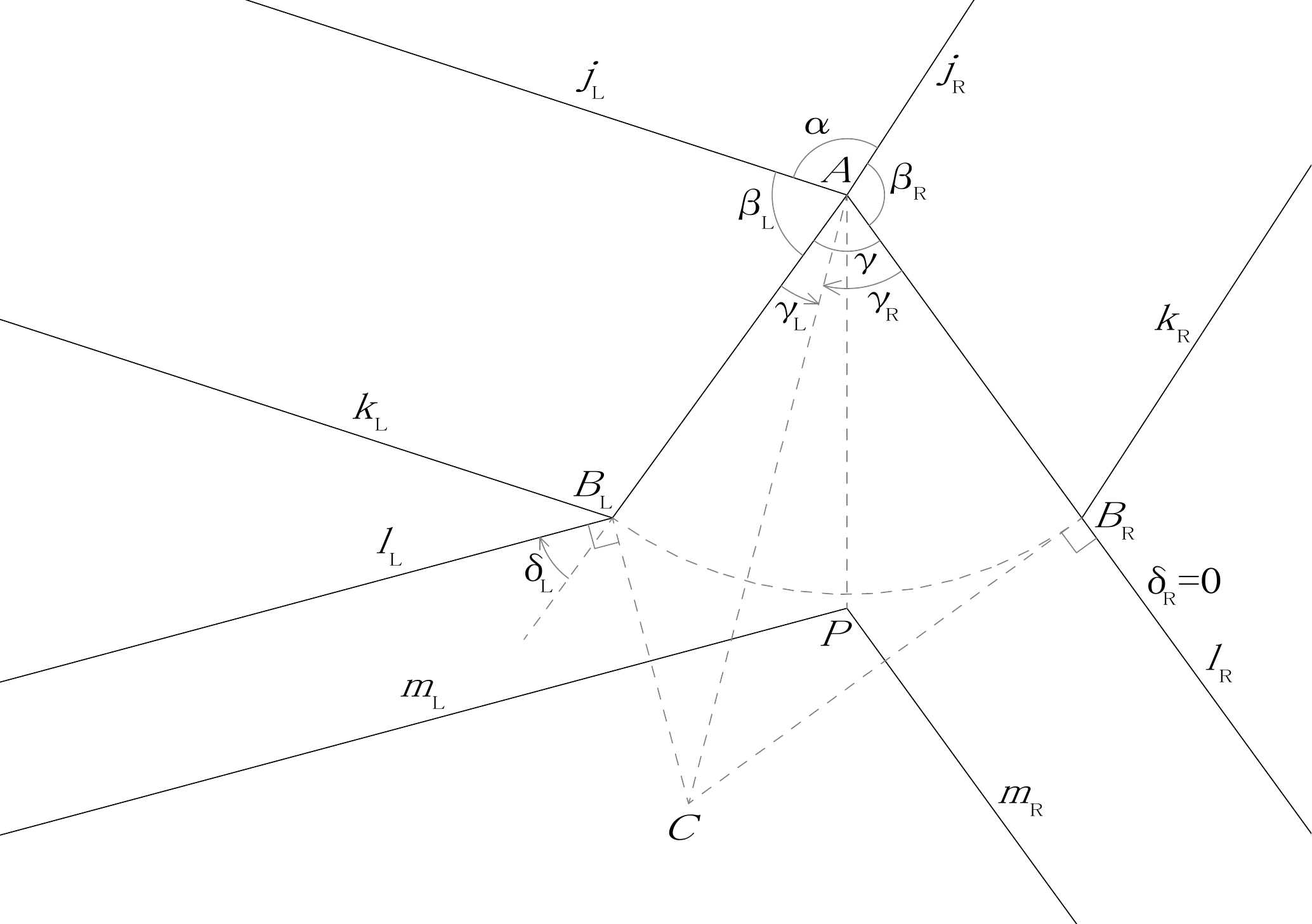}
\caption{Construction of creases of prescribed pleats for an origon gadget}
\label{fig:dev_1}
\end{figure}
\begin{remark}\label{rem:ext_delta}
Define $\gamma_\sigma\in (-\pi /2,\pi /2)$ by $\gamma_\Lt =\angle_+B_\Lt AC$ and $\gamma_\Rt =\angle_+B_\Rt AC$,
so that we have $\gamma_\Lt +\gamma_\Rt =\gamma$.
In our previous papers \cite{Doi19}, \cite{Doi20} and \cite{Doi21},
we assumed $\delta_\sigma <\pi /2$ in addition to $\delta_\sigma <\beta_\sigma$ in condition (iii.b).
This was because angle $\gamma_\sigma$ is given by
\begin{equation}\label{eq:gamma_sigma}
\begin{aligned}
\tan\gamma_\sigma&=\frac{1-\cos\gamma +\sin\gamma\tan\delta_{\sigma'}}{\sin\gamma +\cos\gamma\tan\delta_{\sigma'} +\tan\delta_\sigma},\quad\text{or}\\
\tan\left(\gamma_\sigma -\frac{\gamma}{2}\right)&=\frac{\tan\delta_{\sigma'}-\tan\delta_\sigma}{2+(\tan\delta_\sigma +\tan\delta_{\sigma'})/\tan (\gamma /2)},
\end{aligned}
\end{equation}
and we avoided the appearance of $\tan(\pi /2)$.
However, since we see from condition (iii.c) of Construction $\ref{const:cond}$ that $\delta_\Lt$ and $\delta_\Rt$ do not take $\pi /2$ simultaneously,
we can define $\tan\gamma_\sigma$ continuously even
at $\delta_\sigma =\pi /2$ or $\delta_{\sigma'}=\pi /2$
by the limit of $\tan\gamma_\sigma$ as $\delta_\sigma\to\pi /2$ or $\delta_{\sigma'}\to\pi /2$, so that we have
\begin{equation*}
\gamma_\sigma =\begin{dcases}0&\text{if }\delta_\sigma =\pi /2,\\
\gamma&\text{if }\delta_{\sigma'}=\pi /2\end{dcases}
\end{equation*}
as expected.
Hence in this paper we will not assume $\delta_\sigma <\pi /2$ in condition (iii.b).
Note that $\gamma_\sigma <0$ holds if and only if $\delta_\sigma >\pi /2$,
in which case we have $\gamma_{\sigma'}=\gamma -\gamma_\sigma >\gamma$ and $\delta_{\sigma'}<\pi /2$.
\end{remark}

We end this section with the organization of this paper.
In Section $\ref{sec:redef_crit}$, we give in Definition $\ref{def:crit_geom}$ a redefinition of critical angles originally introduced in \cite{Doi20},
and then recall the construction of positive origon gadgets in Construction $\ref{const:pos}$.
Then after proving the equivalence of the original and the new definition of the critical angles in Theorem $\ref{thm:crit_equiv}$,
we give in Theorem $\ref{thm:zeta_L+R}$ a simplified proof of the existence theorem of positive origons using the new definition of critical angles.
In Section $\ref{sec:const_neg_can_geom}$, we give in Construction $\ref{const:neg_can}$
a geometric (ruler and compass) construction of canonical negative origon gadgets.
The equivalence of the original numerical and the new geometric definition of the canonical dividing point is proved in Theorem $\ref{thm:coinc_D}$.
Then in Section $\ref{sec:exist_pos_can}$, we prove in Theorem $\ref{thm:exist_pos_can}$ the existence of canonical positive origons.
Finally, Section $\ref{sec:concl}$ gives our conclusion.
\section{Redefinition of the critical angles}\label{sec:redef_crit}
First we recall the original definition of the critical angles used in constructing positive origon gadgets,
which is given in \cite{Doi20}, Definition $3.2$.
\begin{definition}\label{def:crit_orig}
Consider a development as in Figure $\ref{fig:dev_1}$, for which we assume conditions (i), (ii) and (iii.a)--(iii.c) of Construction $\ref{const:cond}$.
Then we define the \emph{critical angles} $\zeta_\sigma\in (0,\gamma /2]$ for $\sigma =\Lt ,\Rt$ by the following construction,
where all procedures are done for both $\sigma =\Lt ,\Rt$.
\begin{enumerate}
\item Draw a ray $n_\sigma$ starting from $B_\sigma$ and going inside $\angle B_\Lt AB_\Rt$ so that
$\angle AB_\sigma n_\sigma =\pi -\beta_\sigma +\delta_\sigma <\pi$.
\item Let $Q_\sigma$ be the intersection point of ray $n_\sigma$ and polygonal chain $APm_\sigma$.
\item Then we define the critical angle $\zeta_\sigma$ by
\begin{equation*}
\zeta_\sigma =\angle B_\sigma AQ_\sigma .
\end{equation*}
\item Let $D_\sigma$ be a point in minor arc $\arc{B_\Lt B_\Rt}$ such that $\angle B_\sigma AD_\sigma =2\zeta_\sigma$.
Let us call $D_\sigma$ a \emph{critical dividing point}.
\item The critical angles $\zeta_\Lt ,\zeta_\Rt$ and the critical dividing points $D_\Lt ,D_\Rt$
constructed here are shown in Figure $\ref{fig:const_crit_geom}$,
where the possible range of radius $AD$ of minor arc $\arc{B_\Lt B_\Rt}$ with center $A$ is shaded, which is used in Construction $\ref{const:pos}$, $(1)$.
\end{enumerate}
\end{definition}
\addtocounter{theorem}{1}
\begin{figure}[htbp]
\centering\includegraphics[width=0.75\hsize]{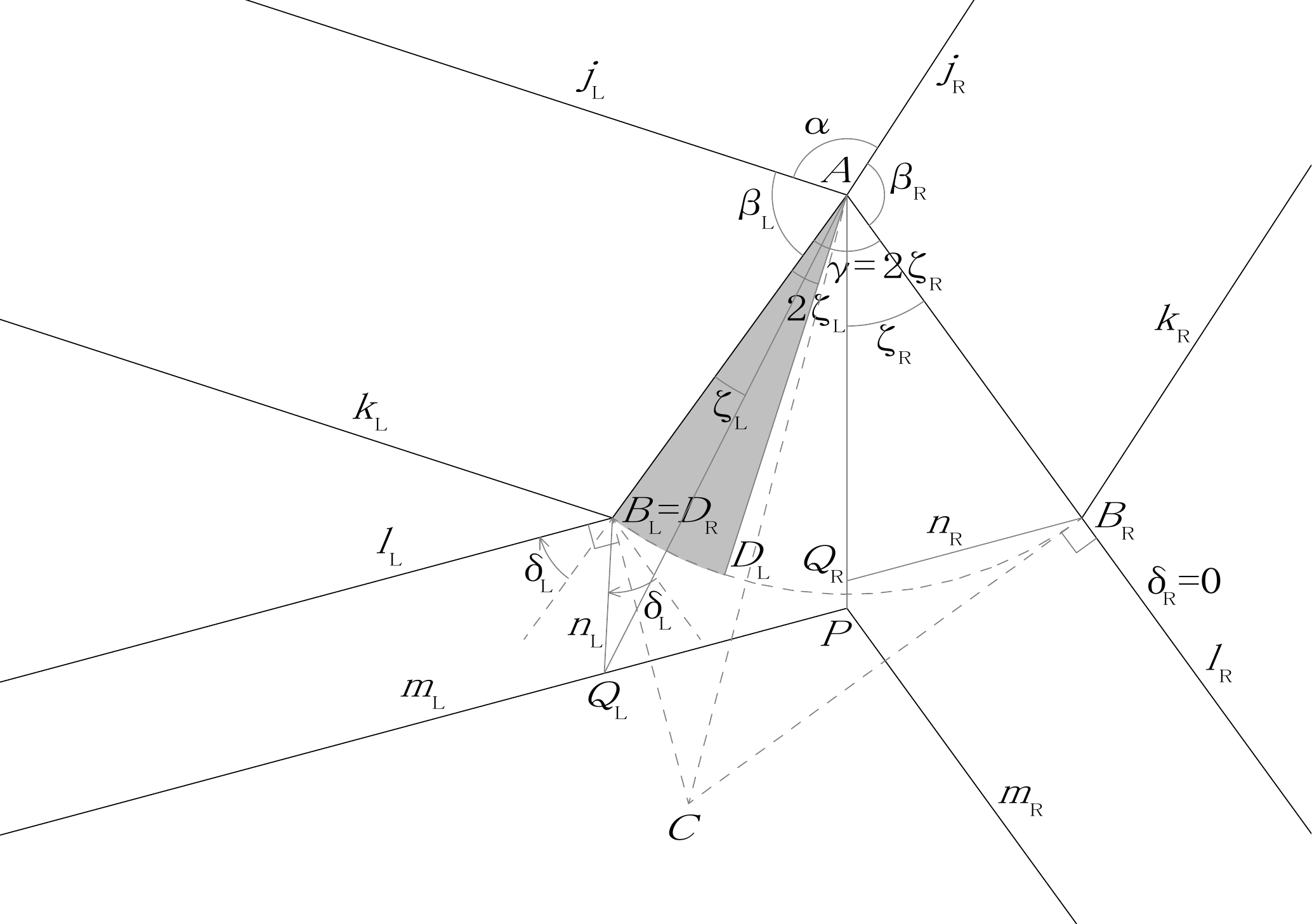}
\caption{Original construction of the critical angles $\zeta_\Lt ,\zeta_\Rt$ and the critical dividing points $D_\Lt ,D_\Rt$}
\label{fig:const_crit_orig}
\end{figure}
The critical angles are also defined numerically by the following result.
\begin{proposition}[\cite{Doi20}, Proposition $3.6$]\label{prop:crit_num}
The critical angle $\zeta_\sigma$ given in Definition $\ref{def:crit_orig}$ is numerically given by
\begin{equation}\label{eq:zeta_num}
\zeta_\sigma =\begin{dcases}
\tan^{-1}\left(\frac{1-d_\sigma /c'}{1/c+1 /c'+(1+d_\sigma /c)/b_\sigma}\right)
&\text{if }\beta_\sigma +\gamma /2+\delta_{\sigma'}\leqslant\pi ,\\
\gamma /2&\text{if }\beta_\sigma +\gamma /2+\delta_{\sigma'}\geqslant\pi ,
\end{dcases}
\end{equation}
where we set
\begin{equation*}
b_\sigma=\tan (\beta_\sigma -\delta_\sigma ),\quad c=\tan (\gamma /2),\quad c'=\tan (\gamma /2+\delta_\Lt +\delta_\Rt ),
\quad\text{and }d_\sigma =\tan\delta_\sigma .
\end{equation*}
In particular, if $\delta_\Lt =\delta_\Rt =0$, then we have $c=c'=\tan (\gamma /2)$ and $d_\Lt =d_\Rt =0$, so that
\begin{equation*}
\zeta_\sigma =\tan^{-1}\left(\frac{1}{2/\tan (\gamma /2)+1/\tan\beta_\sigma}\right) <\frac{\gamma}{2}.
\end{equation*}
\end{proposition}
To introduce a new geometric definition of the critical angles, we shall prepare the following.
\begin{definition}\label{def:angles}
Let $c_A$ be the circle with center $A$ through $B_\Lt$ and $B_\Rt$,
and let $\arc{B_\Lt B_\Rt}$ (resp. $\arc{B_\Rt B_\Lt}$) be the minor (resp. major) arc of circle $c_A$ with endpoints $B_\Lt$ and $B_\Rt$
(see Convention $\ref{conv}$, $(4)$).
We define angle functions $\phi_\sigma$ on minor arc $\arc{B_\Lt B_\Rt}$ with range $[0,\gamma ]$
and $\phi'_\sigma$ on major arc $\arc{B_\Rt B_\Lt}$ with range $[0,2\pi -\gamma ]$ by
\begin{align*}
\phi_\sigma (D)&=\angle B_\sigma AD=
\begin{dcases}\angle_+B_\sigma AD&\text{if }\sigma =\Lt ,\\
\angle_-B_\sigma AD&\text{if }\sigma =\Rt\end{dcases}\quad\text{for }D\in\arc{B_\Lt B_\Rt},\quad\text{and}\\
\phi'_\sigma (D)&=
\begin{dcases}\angle_-B_\sigma AD&\text{if }\sigma =\Lt ,\\
\angle_+B_\sigma AD&\text{if }\sigma =\Rt\end{dcases}\quad\text{for }D\in\arc{B_\Rt B_\Lt}.
\end{align*}
Then we have $\phi_\Lt +\phi_\Rt =\gamma$ on minor arc $\arc{B_\Lt B_\Rt}$ and $\phi'_\Lt +\phi'_\Rt =2\pi -\gamma$ on major arc $\arc{B_\Rt B_\Lt}$.
We can regard $\phi_\sigma$ (resp. $\phi'_\sigma$) as a distance from $B_\sigma$ on minor arc $\arc{B_\Lt B_\Rt}$ (resp. major arc $\arc{B_\Rt B_\Lt}$).
Also, we define angle functions $\psi_\sigma$ with range $[-\gamma_{\sigma'},2\pi -\gamma_{\sigma'}]$ and $\rho_\sigma$ on circle $c_A$ by
\begin{align*}
\psi_\sigma (D)&=\begin{dcases}\angle_-CAD&\text{if }\sigma =\Lt ,\\
\angle_+CAD&\text{if }\sigma =\Rt ,\end{dcases}\quad\text{and}\\
\rho_\sigma (D)&=\angle ACB_\sigma -\angle DCB_\sigma =
\begin{dcases}\angle_+ACD&\text{if }\sigma =\Lt ,\\
\angle_-ACD&\text{if }\sigma =\Rt .\end{dcases}
\end{align*}
Then we have $\phi_\sigma +\psi_\sigma =\gamma_\sigma$ and $\psi_\Lt +\psi_\Rt =0$ on minor arc $\arc{B_\Lt B_\Rt}$, and
$\rho_\sigma (\arc{B_\Lt B_\Rt})=[\gamma_{\sigma'}+\delta_{\sigma'}-\pi /2,\pi /2-\gamma_\sigma -\delta_\sigma ]$.
Note that $\rho_\sigma$ can be expressed in terms of $\psi_\sigma$ by
\begin{equation}\label{eq:rho}
\rho_\sigma =\frac{\sin\psi_\sigma}{r-\cos\psi_\sigma}.
\end{equation}
where $r$ is defined by $r=\norm{AC}/\norm{AB}$ and represented in terms of $\gamma_\sigma$ and $\delta_\sigma$ by
\begin{equation}\label{eq:r}
r=\frac{1}{\cos\gamma_\sigma -\sin\gamma_\sigma\tan\delta_\sigma}\quad\text{for }\delta_\sigma\neq\frac{\pi}{2}.
\end{equation}
Thus if $\delta_\sigma =\pi /2$, then we have $\delta_{\sigma'}\neq 0$ by condition (iii.c) of Construction $\ref{const:cond}$, so that
\begin{equation*}
r=\frac{1}{\cos\gamma -\sin\gamma\tan\delta_{\sigma'}}
\end{equation*}
is well-defined.
The angles defined here are shown in Figure $\ref{fig:angles}$.
\end{definition}
\addtocounter{theorem}{1}
\begin{figure}[htbp]
\centering\includegraphics[width=0.75\hsize]{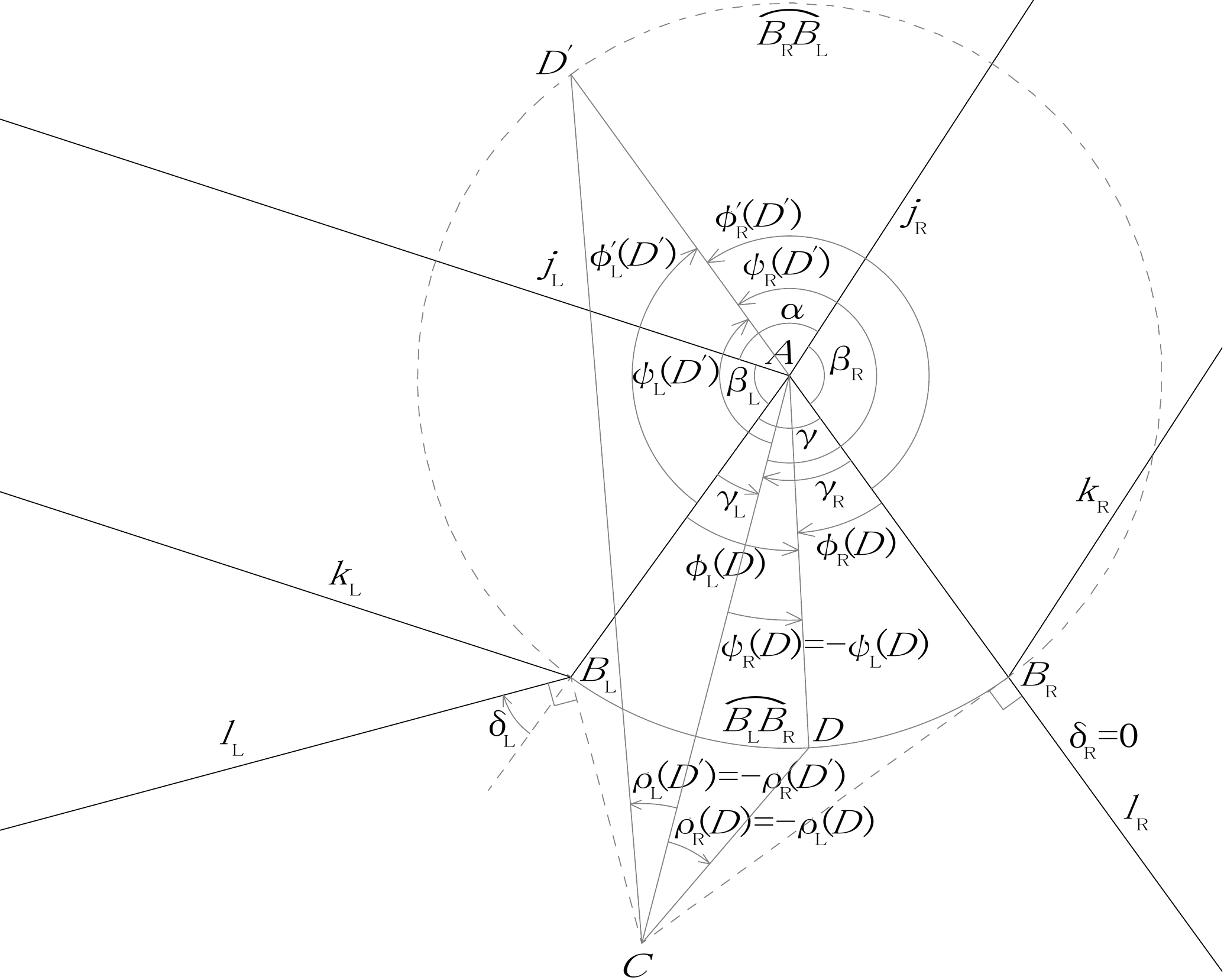}
\caption{Angles given in Definition $\ref{def:angles}$}
\label{fig:angles}
\end{figure}

Now we shall give a new geometric definition of the critical angles.
\begin{definition}\label{def:crit_geom}
We consider a net as in Figure $\ref{fig:dev_0}$.
As in Definition $\ref{def:angles}$, let $c_A$ be a circle with center $A$ through $B_\Lt$ and $B_\Rt$,
$\arc{B_\Lt B_\Rt}$ and $\arc{B_\Rt B_\Lt}$ be a minor and a major arc of $c_A$ respectively, and $\phi_\sigma ,\phi'_\sigma$ be angle functions.
Then we define the critical angles $\zeta_\sigma\in (0,\gamma /2]$ for $\sigma =\Lt ,\Rt$ as follows, where we regard $\sigma$ as taking both $\Lt$ and $\Rt$.
\begin{enumerate}
\item Let $D'_\sigma$ be the intersection point of circle $c_A$ and a perpendicular to $k_\sigma$ through $B_\sigma$.
Equivalently, we can define a point $D'_\sigma$ so that $\phi'_\sigma (D'_\sigma )=2\beta_\sigma$.
Then let $B'$ be the intersection point of segments $B_\Lt D'_\Lt$ and $B_\Rt D'_\Rt$.
\item Let $B'_\sigma$ be the intersection point of major arc $\arc{B_\Rt B_\Lt}$ and an extension of segment $CB_\sigma$.
Alternatively, we can define a point $B'_\sigma$ so that $\phi'_\sigma (D'_\sigma )=2\delta_\sigma$
by $\eqref{eq:B'_sigma}$ in Proposition $\ref{prop:rel_B'_D'}$.
\item If $D'_\sigma$ lies between $B'_\Lt$ and $B'_\Rt$ in major arc $\arc{B_\Rt B_\Lt}$,
then we define $D_\sigma$ to be the intersection point of minor arc $\arc{B_\Lt B_\Rt}$ and segment $CD'_\sigma$.
If $D'_\sigma$ lies between $B_{\sigma'}$ and $B'_{\sigma'}$ in major arc $\arc{B_\Rt B_\Lt}$, then we define $D_\sigma =B_{\sigma'}$.
We call $D_\sigma$ a \emph{critical dividing point}.
(We can also skip procedure $(2)$ and define $D_\sigma$ as the intersection point of minor arc $\arc{B_\Lt B_\Rt}$ and segment $CD'_\sigma$ if it exists,
and $B_{\sigma'}$ otherwise, although the construction of $B_\sigma$ in $(2)$ tells us in advance
whether segment $CD'_\sigma$ intersects minor arc $\arc{B_\Lt B_\Rt}$ or not.
See also Proposition $\ref{prop:rel_B'_D'}$ for the condition that $CD'_\sigma$ intersects $\arc{B_\Lt B_\Rt}$.)
\item Then we define the critical angle $\zeta_\sigma$ by $\zeta_\sigma =\phi_\sigma (D_\sigma )/2$.
\item The critical angles $\zeta_\sigma$ and the critical dividing points $D_\sigma$ constructed here are shown in Figure $\ref{fig:const_crit_geom}$,
where the possible range of radius $AD$ of minor arc $\arc{B_\Lt B_\Rt}$ with center $A$ is shaded, which is used in Construction $\ref{const:pos}$, $(1)$.
\end{enumerate}
\end{definition}
\addtocounter{theorem}{1}
\begin{figure}[htbp]
\centering\includegraphics[width=0.75\hsize]{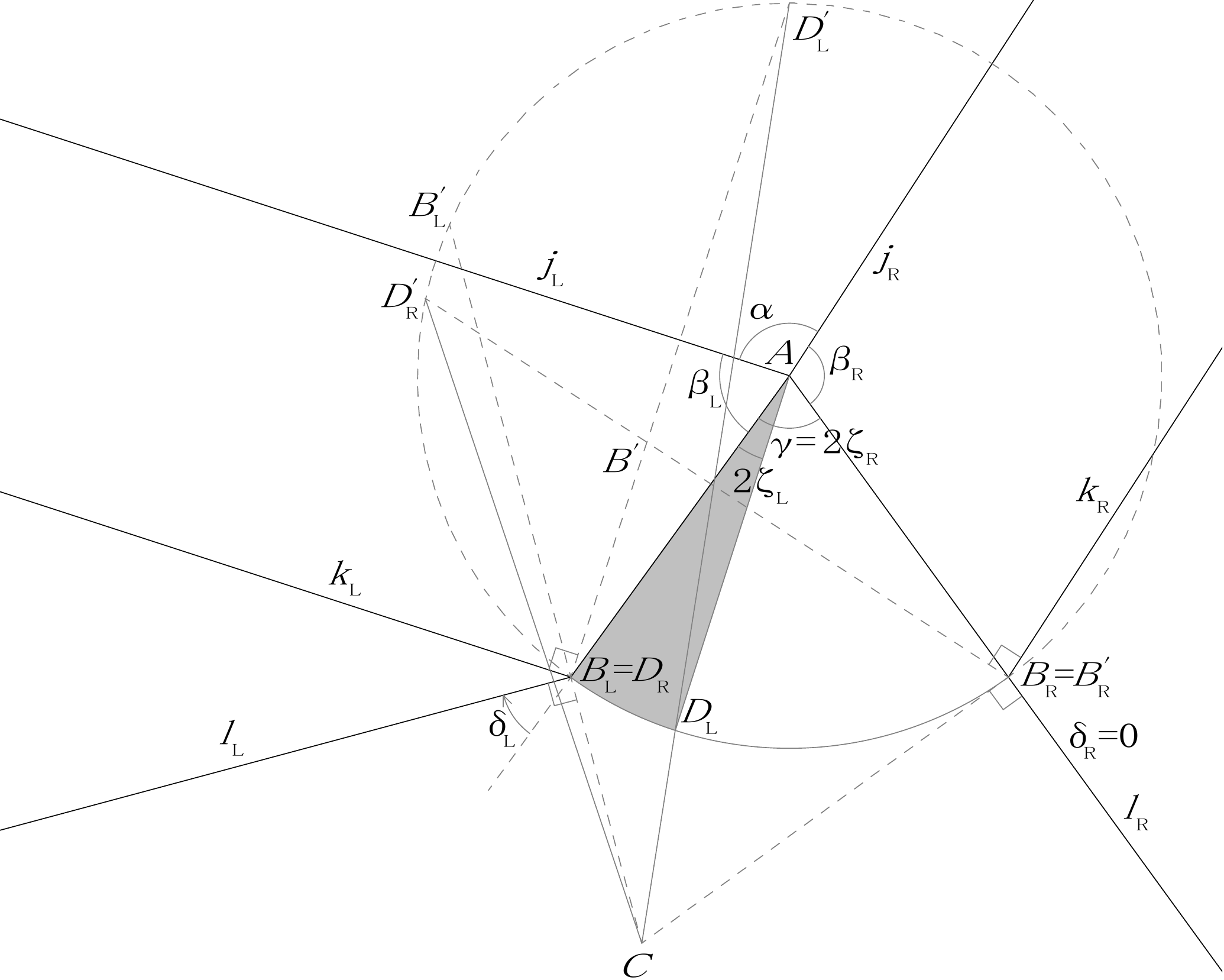}
\caption{New geometric construction of the critical angles $\zeta_\Lt ,\zeta_\Rt$ and the critical dividing points $D_\Lt ,D_\Rt$}
\label{fig:const_crit_geom}
\end{figure}
Definition $\ref{def:crit_geom}$ has advantages over Definition $\ref{def:crit_orig}$ in that
it gives directly the possible range $\arc{D_\Rt D_\Lt}$ of the dividing point $D$ used in Construction $\ref{const:pos}$, $(1)$,
and the construction of $D'_\sigma$ in Definition $\ref{def:crit_orig}$ is simpler than that of $Q_\sigma$ in Definition $\ref{def:crit_orig}$.
We will prove in Theorem $\ref{thm:crit_equiv}$
that the critical angles defined in Definitions $\ref{def:crit_orig}$ and $\ref{def:crit_geom}$ are identical.

Then the construction of positive origon gadgets is given as follows.
\begin{construction}\label{const:pos}
Consider a development as in Figure $\ref{fig:dev_1}$, for which we assume conditions (i), (ii) and (iii.a)--(iii.c) of Construction $\ref{const:cond}$.
Then the crease pattern of a positive origon gadget with prescribed simple outgoing pleats $(\ell_\Lt ,m_\Lt)$ and $(\ell_\Rt ,m_\Rt )$
is constructed as follows, where we regard $\sigma$ as taking both $\Lt$ and $\Rt$.
\begin{enumerate}
\item Let $\zeta_\Lt$ and $\zeta_\Rt$ be the critical angles given in either Definition $\ref{def:crit_orig}$, Proposition $\ref{prop:crit_num}$
or Definition $\ref{def:crit_geom}$.
Alternatively, let $D_\Lt$ and $D_\Rt$ be the points given in either Definition $\ref{def:crit_orig}$ or Definition $\ref{def:crit_geom}$.
Choose a \emph{dividing point} $D$ in minor arc $\arc{B_\Lt B_\Rt}$ with center $A$ so that either of the following conditions holds:
\begin{itemize}
\item $\phi_\sigma =\phi_\sigma (D)=\angle B_\sigma AD$ satisfies
\begin{equation*}
\phi_\Lt\in [\gamma -2\zeta_\Rt ,2\zeta_\Lt ]\cap (0,\gamma ),\quad\text{or equivalently,}\quad\phi_\Rt\in [\gamma -2\zeta_\Lt ,2\zeta_\Rt ]\cap (0,\gamma );
\quad\text{or}
\end{equation*}
\item $D\in\arc{D_\Rt D_\Lt}\setminus\{B_\Lt ,B_\Rt\}$.
\end{itemize}
If $\phi_\sigma =2\zeta_\sigma$, or equivalently, $D=D_\sigma$ for either $\sigma$, then the resulting origon gadget is said to be \emph{critical}.
\item Let $E_\sigma$ be the intersection point of $m_\sigma$ and the bisector of $\angle B_\sigma AD_\can$,
and redefine $m_\sigma$ to be a ray starting from $E_\sigma$ and going in the same direction as $\ell_\sigma$.
\item Determine a point $G_\sigma$ on segment $A E_\sigma$ such that $\angle A B_\sigma G_\sigma=\pi -\beta_\sigma$.
\item If $\delta_\sigma >0$, then determine a point $H_\sigma$ on segment $A E_\sigma$ such that $\angle E_\sigma B_\sigma H_\sigma =\delta_\sigma$.
\item The crease pattern is shown as the solid lines in Figure $\ref{fig:pos_CP}$ for a general case,
and Figure $\ref{fig:pos_crit_CP}$ for a critical case.
Also, the assignment of mountain and valley folds is given in Table $\ref{tbl:pos_assign}$.
\end{enumerate}
\end{construction}
\addtocounter{theorem}{1}
\begin{figure}[htbp]
\centering\includegraphics[width=0.75\hsize]{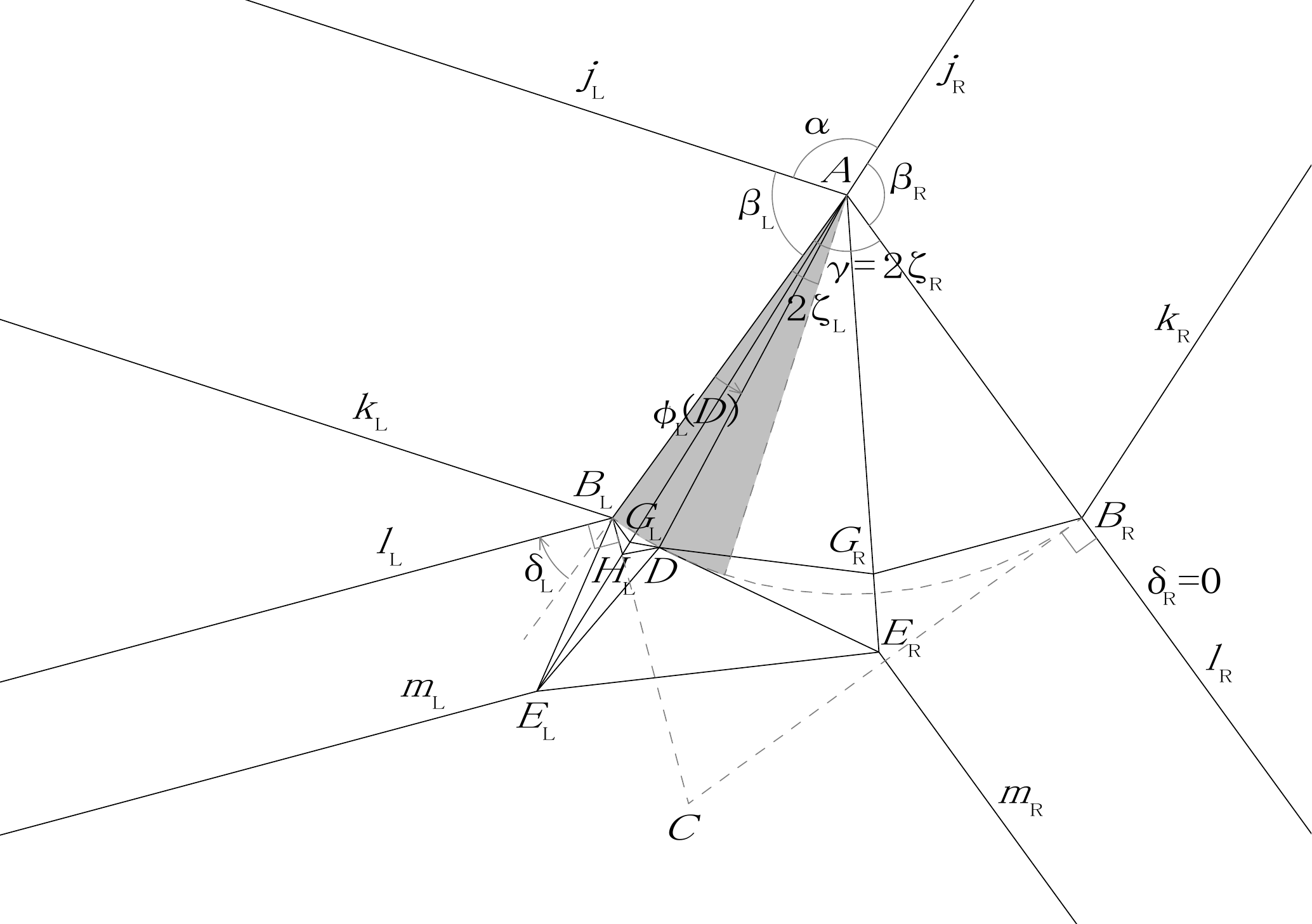}
\caption{Crease pattern of a positive origon gadget, where the possible range of radius $AD$ is shaded}
\label{fig:pos_CP}
\end{figure}
\addtocounter{theorem}{1}
\begin{figure}[htbp]
\centering\includegraphics[width=0.75\hsize]{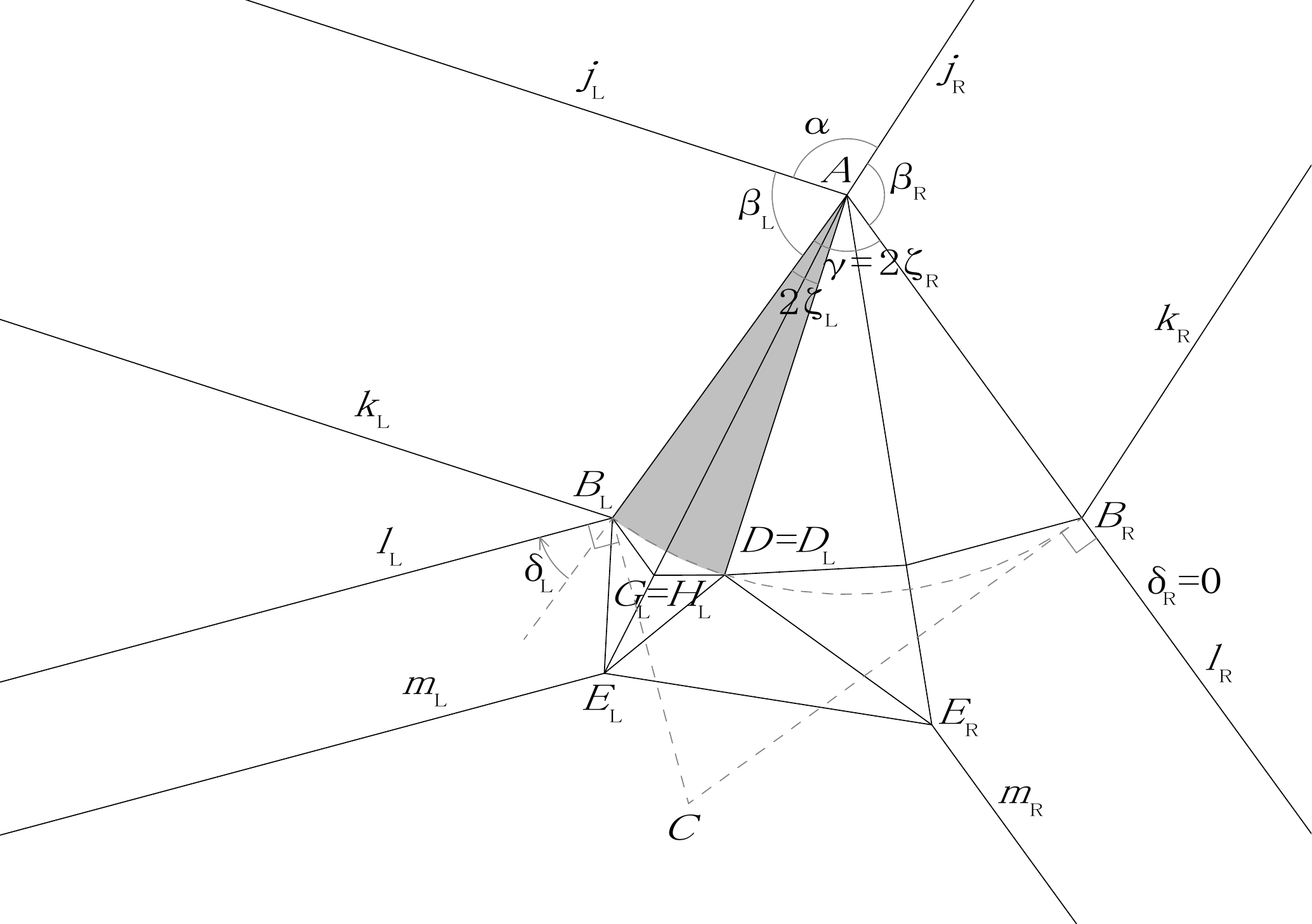}
\caption{Crease pattern of a critical positive origon gadget with $D=D_\Lt$}
\label{fig:pos_crit_CP}
\end{figure}
\renewcommand{\arraystretch}{1.5}
\addtocounter{theorem}{1}
\begin{table}[h]
\begin{tabular}{c|c|c|c|c|c}
&common&\multicolumn{2}{c|}{$\phi_\sigma /2<\zeta_\sigma$ or $D\neq D_\sigma$, and}
&\multicolumn{2}{c}{$\phi_\sigma /2=\zeta_\sigma$ or $D=D_\sigma$, and}\\
\cline{3-6}
&creases&$\delta_\sigma=0$&$\delta_\sigma>0$&$\delta_\sigma=0$&$\delta_\sigma>0$\\
\hline
mountain&$j_\sigma ,\ell_\sigma ,$&\multicolumn{2}{c|}{$B_\sigma G_\sigma$}&$B_\sigma E_\sigma =B_\sigma G_\sigma ,$&$B_\sigma E_\sigma ,$\\
\cline{3-4}
folds&$AB_\sigma ,AD$&$DE_\sigma$&$B_\sigma E_\sigma ,DH_\sigma$&$DE_\sigma =DG_\sigma$&$DG_\sigma =DH_\sigma$\\
\hline
valley&$k_\sigma ,m_\sigma ,$&\multicolumn{2}{c|}{$DG_\sigma$}&\multirow{2}{*}{---}
&\multirow{2}{*}{$B_\sigma G_\sigma =B_\sigma H_\sigma$}\\
\cline{3-4}
folds&$AE_\sigma ,E_\Lt E_\Rt$&---&$B_\sigma H_\sigma$&&
\end{tabular}\vspace{0.5cm}
\caption{Assignment of mountain and valley folds to a positive origon gadget}
\label{tbl:pos_assign}
\end{table}
According to \cite{Doi20}, Setion $4$, the condition that a positive origon gadget is constructible on the side of $\sigma$ is given by
\begin{equation}\label{ineq:constrbil_angle}
\angle AB_\sigma E_\sigma\leqslant\beta_\sigma -\delta_\sigma .
\end{equation}
We can rewrite $\eqref{ineq:constrbil_angle}$ to state the constructibility theorem of positive origons as follows.
\begin{theorem}[\cite{Doi20}, Theorems $4.1$ and $4.3$]\label{thm:constrbil_pos}
Let $D$ be a dividing point chosen in procedure $(1)$ of Construction $\ref{const:pos}$, and
let $\phi_\sigma =\phi_\sigma (D)$, $\psi_\sigma =\psi_\sigma (D)$ and $\rho_\sigma =\rho_\sigma (D)$ be as in Definition $\ref{def:angles}$.
The constructibility condition $\eqref{ineq:constrbil_angle}$ of a positive origon gadget is equivalent to both of the following conditions:
\begin{gather}
\phi_\sigma\leqslant 2\zeta_\sigma ;\quad\text{and}\label{ineq:constrbil_phi}\\
\beta_\sigma +\frac{\gamma_\sigma}{2}+\frac{\psi_\sigma}{2}+\rho_\sigma\geqslant\frac{\pi}{2},\label{ineq:constrbil_psi}
\end{gather}
where $\zeta_\sigma$ is the critical angle given either geometrically by Definition $\ref{def:crit_orig}$ or $\ref{def:crit_geom}$,
or numerically by Proposition $\ref{prop:crit_num}$.
In particular, the equality of $\eqref{ineq:constrbil_phi}$ is equivalent to that of $\eqref{ineq:constrbil_psi}$.
\end{theorem}
\begin{theorem}\label{thm:crit_equiv}
The critical angles $\zeta_\Lt$ and $\zeta_\Rt$ given in Definition $\ref{def:crit_geom}$ are the same as those given in Definition $\ref{def:crit_orig}$.
\end{theorem}
To prove the theorem, we will use the following two results.
\begin{proposition}\label{prop:tan_rho_D'}
Let $D'_\sigma$ be a point in major arc $\arc{B_\Rt B_\Lt}$ with center $A$ given in Definition $\ref{def:crit_geom}$, $(1)$,
so that $\phi'_\sigma (D'_\sigma )=2\beta_\sigma$.
Suppose that there exists a point $D$ in minor arc $\arc{B_\Lt B_\Rt}$ with center $A$ satisfying
\begin{equation}\label{eq:crit_psi}
\beta_\sigma +\frac{\gamma_\sigma}{2}+\frac{\psi_\sigma (D)}{2}+\rho_\sigma (D)=\frac{\pi}{2}.
\end{equation}
Then we have $\rho_\sigma (D)=\rho_\sigma (D'_\sigma )$.
\end{proposition}
\begin{proof}
Let $\psi_\sigma =\psi_\sigma (D)$ and $\rho_\sigma =\rho_\sigma (D)$.
Define $t$ and $\tau$ by
\begin{equation}\label{eq:t_tau}
t=\tan\frac{\psi_\sigma}{2}\quad\text{and}\quad\tau =\tan\left(\frac{\pi}{2}-\beta_\sigma -\frac{\gamma_\sigma}{2}\right)
=\frac{\cos (\beta_\sigma +\gamma_\sigma /2)}{\sin (\beta_\sigma +\gamma_\sigma /2)}.
\end{equation}
Recall from \cite{Doi20}, Theorem $4.3$ that
\begin{equation}\label{eq:rel_t_tau}
\tan\left(\frac{\psi_\sigma}{2}+\rho_\sigma\right) =\frac{r+1}{r-1}\cdot\tan\frac{\psi_\sigma}{2},\quad\text{so that}\quad\tau =\frac{r+1}{r-1}\cdot t.
\end{equation}
Thus if $D$ satisfies $\eqref{eq:crit_psi}$, then we calculate $\tan\rho_\sigma (D)$
using $\eqref{eq:rho}$, $\eqref{eq:t_tau}$ and $\eqref{eq:rel_t_tau}$ as
\begin{align*}
\tan\rho_\sigma (D)&=\frac{\sin\psi_\sigma}{r-\cos\psi_\sigma}=\frac{2t/(1+t^2)}{r-(1-t^2)/(1+t^2)}\\
&=\frac{2t}{(1+t^2)r-(1-t^2)}=\frac{2(r-1)/(r+1)}{(r-1)+(r-1)^2/(r+1)}\\
&=\frac{2\tau}{r+1+(r-1)\tau^2}=\frac{2\tau}{r(1+\tau^2)-(\tau^2-1)}\\
&=\frac{2\tau/(1+\tau^2)}{r-(\tau^2-1)/(1+\tau^2)}=\frac{\sin (2\beta_\sigma +\gamma_\sigma )}{r-\cos (2\beta_\sigma +\gamma_\sigma )}
=\tan\rho_\sigma (D'_\sigma ),
\end{align*}
which gives $\rho_\sigma (D)=\rho_\sigma (D'_\sigma )$ as desired.
This completes the proof of Proposition $\ref{prop:tan_rho_D'}$.
\end{proof}
\begin{proposition}\label{prop:rel_B'_D'}
Let $D'_\sigma$ and $B'_\sigma$ be points in major arc $\arc{B_\Rt B_\Lt}$ with center $A$
given in procedures $(1)$ and $(2)$ of Definition $\ref{def:crit_geom}$ respectively.
Then we have 
\begin{equation}\label{eq:B'_sigma}
\phi'_\sigma (B'_\sigma )=2\delta_\sigma <2\beta_\sigma =\phi'_\sigma (D'_\sigma ).
\end{equation}
Also, either one of the following cases holds for $D'_\sigma$:
\renewcommand{\labelenumi}{$(\mathrm{\roman{enumi}})$}
\begin{enumerate}
\item $\beta_\sigma +\gamma /2+\delta_{\sigma'}=\pi$,
which holds if and only if $\phi'_{\sigma'}(D'_\sigma )=\phi'_{\sigma'}(B'_{\sigma'})=2\delta_{\sigma'}$, or equivalently, $D'_\sigma=B'_{\sigma'}$;
\item $\beta_\sigma +\gamma /2+\delta_{\sigma'}<\pi$, which  holds if and only if 
\begin{equation*}
\phi'_\sigma (B'_\sigma )<\phi'_\sigma (D'_\sigma )<\phi'_\sigma (B'_{\sigma'})
\end{equation*}
or equivalently, $D'_\sigma$ lies strictly between $B'_\Lt$ and $B'_\Rt$ in major arc $\arc{B_\Rt B_\Lt}$; or
\item $\beta_\sigma +\gamma /2+\delta_{\sigma'}>\pi$, which  holds if and only if
\begin{equation*}
0=\phi'_{\sigma'}(B_{\sigma'})<\phi'_{\sigma'}(D'_\sigma )<\phi'_{\sigma'}(B'_{\sigma'}),
\end{equation*}
or equivalently, $D'_\sigma$ lies strictly between $B_{\sigma'}$ and $B'_{\sigma'}$ in major arc $\arc{B_\Rt B_\Lt}$.
\end{enumerate}
\end{proposition}
\begin{proof}
Regarding $\eqref{eq:B'_sigma}$, the second equality is obvious from the construction of $D'_\sigma$ in $(1)$ of Definition $\ref{def:crit_geom}$.
Thus it suffices to prove the first equality because then the middle inequality follows from condition (iii.b) of Construction $\ref{const:cond}$.
First suppose $\delta_\sigma\leqslant\pi /2$, so that $\gamma_\sigma\geqslant 0$.
Then we have
\begin{equation*}
\angle B_\sigma B'_\sigma A=\angle B'_\sigma B_\sigma A=\pi /2-\delta_\sigma ,
\end{equation*}
so that
\begin{equation*}
\phi'_\sigma (B'_\sigma )=\pi -\angle B_\sigma B'_\sigma A-\angle B'_\sigma B_\sigma A=2\delta_\sigma .
\end{equation*}
Next suppose $\delta_\sigma <\pi /2$, so that $\gamma_\sigma <0$.
Then we have
\begin{equation*}
\angle AB'_\sigma B_\sigma =\angle AB_\sigma B'_\sigma =\pi -\angle AB_\sigma C=\pi -(2\pi -(\pi/2 +\delta_\sigma ))=\delta_\sigma -\pi /2,
\end{equation*}
so that
\begin{equation*}
\phi'_\sigma (B'_\sigma )=2\pi -\angle B_\sigma AB'_\sigma =2\pi -(\pi -2(\delta_\sigma -\pi /2))=2\delta_\sigma .
\end{equation*}
Thus we proved $\eqref{eq:B'_sigma}$.
Also, noting that $\phi'_\sigma (D'_\sigma )=2\beta_\sigma$, we see that $\phi'_\sigma (B'_{\sigma'})\lesseqqgtr\phi'_\sigma (D'_\sigma )$ is equivalent to
\begin{align*}
2\delta_{\sigma'}&=\phi'_{\sigma'}(B'_{\sigma'})\gtreqqless\phi'_{\sigma'}(D'_\sigma )\\
&=(2\pi -\gamma )-\phi'_\sigma (D'_\sigma )=2\pi -\gamma -2\beta_\sigma ,
\end{align*}
which gives $\beta_\sigma +\gamma /2+\delta_{\sigma'}\gtreqqless\pi$.
Combining $\eqref{eq:B'_sigma}$ and $\phi'_\sigma +\phi'_{\sigma'}=2\pi -\gamma$, we obtain the equivalent conditions for $D'_\sigma$ in cases (i)--(iii).
This completes the proof.
\end{proof}
\begin{proof}[Proof of Theorem $\ref{thm:crit_equiv}$]
Here let $\zeta_\sigma$ be as in Definition $\ref{def:crit_orig}$,
so that $\zeta_\sigma$ is given numerically by $\eqref{eq:zeta_num}$ of Proposition $\ref{prop:crit_num}$.
Also, let $D_\sigma$ be as in Definition $\ref{def:crit_geom}$.
Then we shall prove that $\phi_\sigma (D_\sigma )=2\zeta_\sigma$.

First suppose $\beta_\sigma +\gamma /2+\delta_{\sigma'}\geqslant\pi$.
Then Proposition $\ref{prop:crit_num}$ gives $\zeta_\sigma =\gamma /2$,
while Proposition $\ref{prop:rel_B'_D'}$ gives $D_\sigma =B_{\sigma'}$, so that $\phi_\sigma (D_\sigma )=\gamma$.
Thus we have $\phi_\sigma (D_\sigma )=2\zeta_\sigma$ as desired.

Next suppose $\beta_\sigma +\gamma /2+\delta_{\sigma'}<\pi$.
Then we see from Proposition $\ref{prop:rel_B'_D'}$ that $D'_\sigma$ lies strictly between $B'_\Lt$ and $B'_\Rt$ in major arc $\arc{B_\Rt B_\Lt}$,
so that $\phi_\sigma (D_\sigma )\in (0,\gamma )$.
Meanwhile, Definition $\ref{def:crit_orig}$ or Proposition $\ref{prop:crit_num}$ gives that $\zeta_\sigma\in (0,\gamma /2)$.
Now let $D$ be a point in minor arc $\arc{B_\Lt B_\Rt}$ such that $\phi_\sigma (D)=2\zeta_\sigma\in (0,\gamma )$.
Then it follows from Theorem $\ref{thm:constrbil_pos}$ that $D$ satisfies
\begin{equation*}
\beta_\sigma +\frac{\gamma_\sigma}{2}+\frac{\psi_\sigma (D)}{2}+\rho_\sigma (D)=\frac{\pi}{2},
\end{equation*}
so that we have $\rho_\sigma (D)=\rho_\sigma (D'_\sigma )$ by Proposition $\ref{prop:tan_rho_D'}$.
Since $D_\sigma$ also satisfies $\rho_\sigma (D_\sigma )=\rho_\sigma (D'_\sigma )$,
we must have $D_\sigma =D$ from the monotonicity of $\rho_\sigma$ on minor arc $\arc{B_\Lt B_\Rt}$.
Hence we have $\phi (D_\sigma )=2\zeta_\sigma$ as desired.

This completes the proof of Theorem $\ref{thm:crit_equiv}$.
\end{proof}

Now we recall the following existence theorem of positive origons.
\begin{theorem}[\cite{Doi20}, Theorem $4.6$]\label{thm:zeta_L+R}
The critical angles $\zeta_\Lt$ and $\zeta_\Rt$ given in either Definition $\ref{def:crit_orig}$, Proposition $\ref{prop:crit_num}$
or Definition $\ref{def:crit_geom}$ always satisfy $\zeta_\Lt +\zeta_\Rt >\gamma /2$.
Equivalently, the critical dividing points $D_\Lt$ and $D_\Rt$ given in Definition $\ref{def:crit_orig}$ or $\ref{def:crit_geom}$
always satisfy $0\leqslant\phi_\Lt (D_\Rt )<\phi_\Lt (D_\Lt )\leqslant\gamma$
(and also $0\leqslant\phi_\Rt (D_\Lt )<\phi_\Rt (D_\Rt )\leqslant\gamma$).
Thus intervals $[\phi_\Lt (D_\Rt ),\phi_\Lt (D_\Lt )]\cap (0,\gamma )=[\gamma -2\zeta_\Rt ,2\zeta_\Lt ]\cap (0,\gamma )$ and 
$[\phi_\Lt (D_\Rt ),\phi_\Lt (D_\Lt )]\cap (0,\gamma )=[\gamma -2\zeta_\Lt ,2\zeta_\Rt ]\cap (0,\gamma )$
given in procedure $(1)$ of Construction $\ref{const:pos}$ include infinitely many points.
\end{theorem}
Here we shall give a geometric proof of this theorem using Definition $\ref{def:crit_geom}$, which simplifies the two proofs given in \cite{Doi20}.
\begin{proof}
It suffices to prove that $\phi_\Lt (D_\Rt )<\phi_\Lt (D_\Lt )$.
Since we see from $(1)$ of Definition $\ref{def:crit_geom}$ that $\phi'_\Lt (D'_\Lt )=2\beta_\Lt$ and $\phi'_\Rt (D'_\Rt )=2\beta_\Rt$, we have
\begin{equation}\label{ineq:D'_R<D'_L}
\phi'_\Lt (D'_\Lt )-\phi'_\Lt (D'_\Rt )=2\beta_\Lt -(2\pi -\gamma -2\beta_\Rt )=2(\beta_\Lt +\beta_\Rt +\gamma /2 -\pi )>0
\end{equation}
by condition (i) of Construction $\ref{const:cond}$.

Now suppose $\beta_\Lt +\gamma /2+\delta_\Rt\geqslant\pi$.
Then by Proposition $\ref{prop:rel_B'_D'}$, we have $\phi_\Lt (D_\Lt )=\gamma$ and $\phi_\Lt (D_\Rt )=\gamma -\phi_\Rt (D_\Rt )\in [0,\gamma )$,
so that $\phi_\Lt (D_\Rt )<\phi_\Lt (D_\Lt )$.
Next suppose $\beta_\Rt +\gamma /2+\delta_\Lt\geqslant\pi$.
Then we also have $0=\phi_\Lt (D_\Rt )<\phi_\Lt (D_\Lt )$ in the same way.
Finally, suppose $\beta_\sigma +\gamma /2+\delta_{\sigma'}<\pi$ for both $\sigma =\Lt ,\Rt$.
Then we have $\phi'_\Lt (B'_\Lt )<\phi'_\Lt (D'_\Rt )<\phi'_\Lt (D'_\Lt )<\phi'_\Lt (B'_\Rt )$ by $\eqref{ineq:D'_R<D'_L}$ and
Proposition $\ref{prop:rel_B'_D'}$, (ii), which leads to $0<\phi_\Lt (D_\Rt )<\phi_\Lt (D_\Lt )<\gamma$.
This completes the proof of Theorem $\ref{thm:zeta_L+R}$.
\end{proof}
\section{Geometric construction of canonical negative origon gadgets}\label{sec:const_neg_can_geom}
We begin with a geometric (ruler and compass) construction of negative origon gadgets.
\begin{construction}\label{const:neg_can}\rm
Consider a development as in Figure $\ref{fig:dev_1}$, for which we require condition (i), (ii) and (iii.a)--(iii.c) of Construction $\ref{const:cond}$,
where the development is seen from the back side according to Convention $\ref{conv}$, $(3)$.

Then the crease pattern of a canonical negative origon gadget with prescribed simple outgoing pleats $(\ell_\Lt ,m_\Lt )$ and $(\ell_\Rt ,m_\Rt )$
is constructed as follows, where we regard $\sigma$ as taking both $\Lt$ and $\Rt$.
\renewcommand{\labelenumi}{(\arabic{enumi})}
\begin{enumerate}
\item Determine the \emph{canonical dividing point} $D_\can$ in minor arc $\arc{B_\Lt B_\Rt}$ with center $A$
either geometrically here or numerically in $(1')$.
Let $B'$ be the intersection point of a perpendicular through $B_\Lt$ to $k_\Lt$ and a perpendicular through $B_\Rt$ to $k_\Rt$.
Then $D_\can$ is determined as the intersection point of segment $B'C$ and minor arc $\arc{B_\Lt B_\Rt}$.
\item[(1${}'$)] Determine the canonical dividing point $D_\can$ in minor arc $\arc{B_\Lt B_\Rt}$ with center $A$
such that $\rho_\Lt =\angle B_\Lt CA -\angle B_\Lt CD_\can$ and $\psi_\Lt =\angle B_\Lt AC -\angle B_\Lt AD_\can$ are calculated
using either Theorem $\ref{thm:exist_rho}$ or equations $\eqref{eq:angle_BCD}$, $\eqref{eq:rho_alt_2}$ in Proposition $\ref{prop:angle_BCD}$,
or alternatively, $\angle B_\Lt CD_\can$ is calculated using $\eqref{eq:angle_BCD}$ in Proposition $\ref{prop:angle_BCD}$.
\setcounter{enumi}{1}
\item Let $E_\sigma$ be the intersection point of $m_\sigma$ and the bisector of $\angle B_\sigma AD_\can$,
and redefine $m_\sigma$ to be a ray starting from $E_\sigma$ and going in the same direction as $\ell_\sigma$.
\item Draw a parallel to segment $E_\Lt E_\Rt$ through $D_\can$, letting $G'_\sigma$ be the intersection point of the parallel and segment $AE_\sigma$.
\item Draw a parallel to segment $E_\Lt E_\Rt$ through $C$, letting $P_\sigma$ be the intersection point of the parallel and ray $m_\sigma$.
\item The desired crease pattern is shown as the solid lines in Figure $\ref{fig:neg_can_CP}$,
and the assignment of mountain and valley folds is given in Table $\ref{tbl:neg_can_assign}$.
\end{enumerate}
\end{construction}
\addtocounter{theorem}{1}
\begin{figure}[htbp]
\centering\includegraphics[width=0.75\hsize]{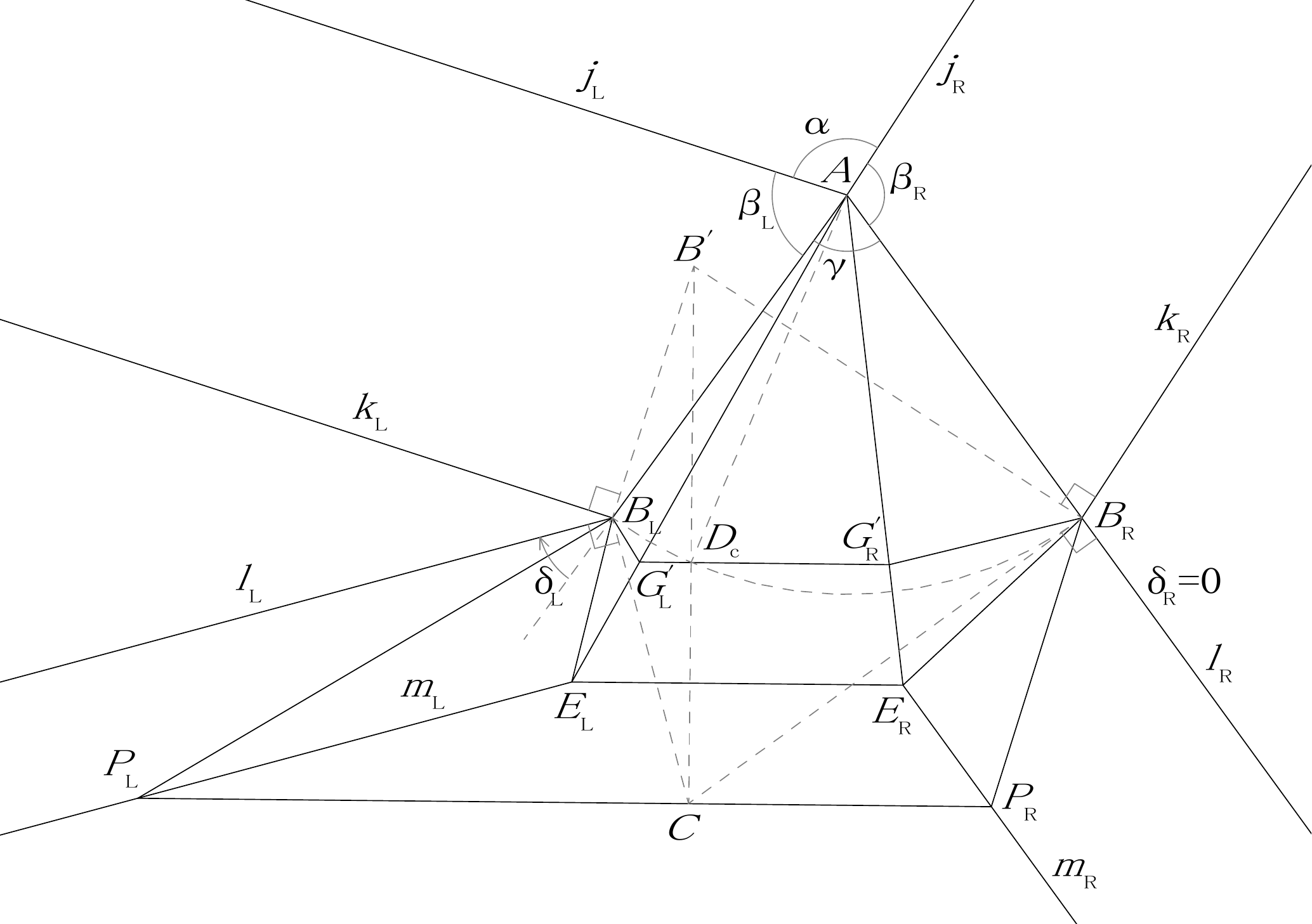}
\caption{Crease pattern of a canonical negative origon gadget}
\label{fig:neg_can_CP}
\end{figure}
\addtocounter{theorem}{1}
\begin{table}[h]
\begin{tabular}{c|c}
mountain folds&$j_\sigma ,m_\sigma ,AE_\sigma ,B_\sigma G'_\sigma ,B_\sigma P_\sigma, E_\Lt E_\Rt$\\ \hline
valley folds&$k_\sigma ,\ell_\sigma ,AB_\sigma ,B_\sigma E_\sigma ,G'_\Lt G'_\Rt ,P_\Lt P_\Rt$
\end{tabular}\vspace{0.5cm}
\caption{Assignment of mountain and valley folds to a canonical negative origon gadget}
\label{tbl:neg_can_assign}
\end{table}
Note that this compass and ruler construction can be also performed easily by origami folds
if we are given a flat piece of paper with creases as in Figure $\ref{fig:dev_0}$.

Now let us find an equation which $\psi_\Lt =\psi_\Lt (D_\can )$ or $\rho_\Lt =\rho_\Lt (D_\can )$ in Construction $\ref{const:neg_can}$ satisfies.
In the resulting extrusion, noting that ridge $AB$ overlaps with $AD_\can$, we need
\begin{equation}\label{eq:ADG=ABG}
\angle AD_\can G_\Rt =\angle ABG_\Rt .
\end{equation}
This is rewritten as
\begin{equation}\label{eq:DA_EE=BA_EE}
e(\ora{D_\can A})\cdot e(\ora{E_\Lt E_\Rt})=e(\ora{BA})\cdot e(\ora{E_\Lt E_\Rt}),
\end{equation}
where where $e(\ora{v})$ denotes a unit direction vector of $\ora{v}$, i.e., $e(\ora{v})=\ora{v}/\norm{v}$, and
$\ora{BA}$ is a $3$D vector in the resulting extrusion while $\ora{D_\can A}$ is a $2$D vector in the development.
Finally, we find a desired equation of $\rho_\Lt =\rho_\Lt (D_\can )$ as follows.
\begin{theorem}[\cite{Doi21}, Section $3.3$]\label{thm:eq_neg_can}
Equation $\eqref{eq:DA_EE=BA_EE}$ (and also equation $\eqref{eq:ADG=ABG}$) is equivalent to
\begin{equation}\label{eq:DA_EE=BA_EE_alt_1}
(V_1-rV_2)\cdot\tan\rho_\Lt =W,
\end{equation}
where $r$ is given by $\eqref{eq:r}$, and $V_1,V_2,W$ are given by
\begin{align*}
V_1&=\sin (\beta_\Lt +\gamma_\Lt )\cos\beta_\Rt +\sin (\beta_\Rt +\gamma_\Rt )\cos\beta_\Lt ,\\
V_2&=\sin(\beta_\Lt +\beta_\Rt +\gamma ),\\
W&=\cos (\beta_\Lt +\gamma_\Lt )\cos\beta_\Rt -\cos (\beta_\Rt +\gamma_\Rt )\cos\beta_\Lt .
\end{align*}
\end{theorem}
We have the existence and uniquness of a solution of $\eqref{eq:DA_EE=BA_EE_alt_1}$ as follows.
\begin{theorem}[\cite{Doi21}, Theorem $3.26$]\label{thm:exist_rho}
Suppose $\alpha ,\beta_\Lt ,\beta_\Rt ,\delta_\Lt$ and $\delta_\Rt$ satisfy
conditions $(\mathrm{i})$, $(\mathrm{ii})$ and $(\mathrm{iii.a})$--$(\mathrm{iii.c})$ of Construction $\ref{const:cond}$.
Then there exists a solution $\rho_\Lt\in (\gamma_\Rt +\delta_\Rt -\pi /2, \pi /2-\gamma_\Lt -\delta_\Lt )$
of equation $\eqref{eq:DA_EE=BA_EE_alt_1}$, which is given by
\begin{equation*}
\rho_\Lt =\tan^{-1}\left(\frac{W}{V_1-rV_2}\right) .
\end{equation*}
Thus there exists a unique point $D_\can$ in open minor arc $\arc{B_\Lt B_\Rt}\setminus\{ B_\Lt ,B_\Rt\}$ with center $A$ such that 
$\rho_\Lt =\rho_\Lt (D_\can )$, and $\psi_\Lt =\psi_\Lt (D_\can )\in (-\gamma_\Rt ,\gamma_\Lt )$ given by
\begin{equation*}
\psi_\Lt =\sin^{-1}(r\sin\rho_\Lt )-\rho_\Lt
\end{equation*}
is a unique solution of $\eqref{eq:ADG=ABG}$.
\end{theorem}
\begin{lemma}\label{lem:B'C_perp_GG}
In the crease pattern resulting in Construction $\ref{const:neg_can}$ using $D_\can$ given in procedure $(1)$,
segment $G'_\Lt G'_\Rt$ is perpendicular to segment $B'C$.
\end{lemma}
\begin{proof}
This is obvious because $E_\Lt E_\Rt$ is a perpendicular bisector of $CD_\can$ and $G'_\Lt G'_\Rt$ is parallel to $E_\Lt E_\Rt$.
\end{proof}
Now we prove the equivalence of the geometric and the numerical definition of the canonical dividing point $D_\can$.
\begin{theorem}\label{thm:coinc_D}
In Construction $\ref{const:neg_can}$, the canonical dividing point $D_\can$ given geometrically in procedure $(1)$ is always constructible,
and is the same as that given numerically in procedure $(1')$.
\end{theorem}
\begin{proof}
Since minor arc $\arc{B_\Lt B_\Rt}$ with center $A$ is contained in quadrilateral $B'B_\Lt CB_\Rt$, segment $B'C$ intersects minor arc $\arc{B_\Lt B_\Rt}$,
which yields the constructiblity of $D_\can$ in $(1)$.

Let $\proj_{A,0}$ be the orthogonal projection of the resulting extrusion to the crease pattern
in the bottom plane $\{ z=0\}$ which fixes $A$ and preserves the angles of the bottom edges.
Then in the process of folding the gadget, the projection $\proj_{A,0}(D_\can )$ of point $D_\can$
starts from $D_\can$ and arrives at $B'$ in segment $D_\can B'$
as segment $AD_\can$ rotates around $G'_\Lt G'_\Rt$ because $D_\can B'$ is perpendicular to $G'_\Lt G'_\Rt$ by Lemma $\ref{lem:B'C_perp_GG}$.
Thus $D_\can$ satisfies equation $\eqref{eq:ADG=ABG}$.
By the uniqueness of the solution of $\eqref{eq:ADG=ABG}$ given in Theorem $\ref{thm:exist_rho}$,
this $D_\can$ determined in $(1)$ coincides with that in $(1')$.
This completes the proof of Theorem $\ref{thm:coinc_D}$.
\end{proof}
Thus we have the following expression of $\angle B_\Lt CD_\can$, which also gives $\rho_\Lt$ in procedure $(1')$ of Construction $\ref{const:neg_can}$.
\begin{proposition}\label{prop:angle_BCD}
We calculate $\angle B_\Lt CD_\can$ and $\angle B_\Rt CD_\can$ as
\begin{equation}\label{eq:angle_BCD}
\begin{aligned}
\tan\angle B_\Lt CD_\can&=\tan\angle B_\Lt CB'=\frac{\sin (\beta_\Lt -\delta_\Lt )}{\displaystyle\frac{\sin\alpha}{\sin (\beta_\Rt +\gamma /2)}
\cdot\frac{\sin (\gamma /2+\delta_\Rt )}{\sin (\gamma +\delta_\Lt +\delta_\Rt )}+\cos (\beta_\Lt -\delta_\Lt )},\\
\tan\angle B_\Rt CD_\can&=\tan\angle B_\Rt CB'=\frac{\sin (\beta_\Rt -\delta_\Rt )}{\displaystyle\frac{\sin\alpha}{\sin (\beta_\Lt +\gamma /2)}
\cdot\frac{\sin (\gamma /2+\delta_\Lt )}{\sin (\gamma +\delta_\Lt +\delta_\Rt )}+\cos (\beta_\Rt -\delta_\Rt )}.
\end{aligned}
\end{equation}
Thus using $\eqref{eq:angle_BCD}$, we can calculate $\rho_\Lt$ and $\rho_\Rt$ by
\begin{equation}\label{eq:rho_alt_2}
\rho_\Lt =\frac{\pi}{2}-(\gamma_\Lt +\delta_\Lt +\angle B_\Lt CD_\can ),\quad\rho_\Rt =(\gamma_\Rt +\delta_\Rt +\angle B_\Rt CD_\can )-\frac{\pi}{2}.
\end{equation}
\end{proposition}
\begin{proof}
For $\sigma =\Lt ,\Rt$, let
\begin{gather*}
b_\sigma =\norm{B_\sigma B'},\quad c_\sigma =\norm{B_\sigma C},\quad d=\norm{B_\Lt B_\Rt}\quad\text{and}\\
\varphi_\sigma =\angle B_\sigma B'C,\quad\theta_\sigma =\angle B_\sigma CB'.
\end{gather*}
Noting that $\angle B'B_\sigma C=\pi -(\beta_\sigma -\delta_\sigma )$, so that $\varphi_\sigma +\theta_\sigma =\beta_\sigma -\delta_\sigma$, we have
\begin{align*}
b_\sigma :c_\sigma&=\sin\theta_\sigma :\sin\varphi_\sigma =\sin\theta_\sigma :\sin (\beta_\sigma -\delta_\sigma -\theta_\sigma )\\
&=\sin\theta_\sigma :\sin (\beta_\sigma -\delta_\sigma )\cos\theta_\sigma -\cos (\beta_\sigma -\delta_\sigma )\sin\theta_\sigma\\
&=1:\sin (\beta_\sigma -\delta_\sigma )/\tan\theta_\sigma -\cos (\beta_\sigma -\delta_\sigma ),
\end{align*}
which gives that
\begin{equation}\label{eq:tan_theta}
\tan\theta_\sigma =\frac{\sin (\beta_\sigma -\delta_\sigma )}{c_\sigma /b_\sigma +\cos (\beta_\sigma -\delta_\sigma )}.
\end{equation}
On the other hand, since we see easily that
\begin{gather*}
\angle B_\Lt B'B_\Rt =\pi -\alpha ,\quad\angle B'B_\sigma B_{\sigma'}=\beta_\sigma +\gamma /2,\quad\text{and}\\
\angle B_\Lt CB_\Rt =\pi -(\gamma +\delta_\Lt +\delta_\Rt ),\quad \angle CB_\sigma B_{\sigma'}=\gamma /2+\delta_\sigma ,
\end{gather*}
we have
\begin{align*}
b_\sigma :d&=\sin\angle B'B_{\sigma'}B_\sigma :\sin\angle B_\Lt B'B_\Rt=\sin (\beta_{\sigma'}+\gamma /2) :\sin\alpha ,\\
c_\sigma :d&= \sin\angle CB_{\sigma'}B_\sigma :\sin\angle B_\Lt CB_\Rt =\sin (\gamma /2+\delta_{\sigma'}):\sin (\gamma +\delta_\Lt +\delta_\Rt ),
\end{align*}
so that
\begin{equation}\label{eq:c/b}
\frac{c_\sigma}{b_\sigma}=\frac{d}{b_\sigma}\cdot\frac{c_\sigma}{d}
=\frac{\sin\alpha}{\sin (\beta_{\sigma'}+\gamma /2)}\cdot\frac{\sin (\gamma /2+\delta_{\sigma'})}{\sin (\gamma +\delta_\Lt +\delta_\Rt )}
\end{equation}
Then combining $\eqref{eq:tan_theta}$ and $\eqref{eq:c/b}$ gives $\eqref{eq:angle_BCD}$ as desired.
This completes the proof of Proposition $\ref{prop:angle_BCD}$.
\end{proof}
\section{Existence of the canonical positive origon gadgets}\label{sec:exist_pos_can}
It is natural to ask whether there exists a positive origon gadget with the same dividing point as a canonical negative one.
The following result gives an affirmative answer to this question.
\begin{theorem}\label{thm:exist_pos_can}
Let $D_\can$ be the canonical dividing point given geometrically in procedure $(1)$ or numerically in procedure $(1')$ of Construction $\ref{const:neg_can}$.
Also, for $\sigma =\Lt ,\Rt$ let $D_\sigma$ be the point given geometrically in Definition $\ref{def:crit_orig}$ or $\ref{def:crit_geom}$,
or given by $\phi_\sigma (D_\sigma )=2\zeta_\sigma$, where $\zeta_\sigma$ is given in Proposition $\ref{prop:crit_num}$.
Then we have $\phi_\Lt (D_\Rt)<\phi_\Lt (D_\can )<\phi_\Lt (D_\Lt)$.
Thus we obtain a positive origon gadget with $D=D_\can$ by Construction $\ref{const:pos}$.
\end{theorem}
\begin{proof}
Let $B'_\Lt$ and $B'_\sigma$ be as in Definition $\ref{def:crit_geom}$.
Note that $B'$ does not coincide with $D'_\Lt$ or $D'_\Rt$ because if $B'=D'_\sigma$ for either $\sigma$, then we have $D'_\Lt =D'_\Rt$,
which contradicts Theorem $\ref{thm:constrbil_pos}$.
Thus it suffices to prove that $\phi_\Lt (D_\Rt )<\phi_\Lt (D_\can )$.
Indeed, by a symmetric argument interchanging $\Lt$ and $\Rt$, we also find that $\phi_\Rt (D_\Lt )<\phi_\Rt (D_\can )$,
which is equivalent to $\phi_\Lt (D_\can )<\phi_\Lt (D_\Lt )$.

By Proposition $\ref{prop:rel_B'_D'}$, we consider the two cases $\beta_\Rt +\gamma /2+\delta_\Lt\geqslant\pi$ and $\beta_\Rt +\gamma /2+\delta_\Lt <\pi$,
which are equivalent to $\phi'_\Lt (D'_\Rt)\leqslant\phi'_\Lt (B'_\Lt )$ and $\phi'_\Lt (B'_\Lt )<\phi'_\Lt (D'_\Rt)$ respectively.
On the other hand, we see that $\phi'_\Lt (B'_\Lt )<\phi'_\Lt (D'_\Lt )$ and $\phi'_\Lt (D'_\Rt )<\phi'_\Lt (B'_\Rt )$,
and $\phi'_\Lt (D'_\Rt )<\phi'_\Lt (D'_\Lt )$ always hold
from Proposition $\ref{prop:rel_B'_D'}$ and Theorem $\ref{thm:constrbil_pos}$ respectively.

First suppose $\phi'_\Lt (D'_\Rt)\leqslant\phi'_\Lt (B'_\Lt )<\phi'_\Lt (D'_\Lt )$.
Then we have $0=\phi_\Lt (D_\Rt )<\phi_\Lt (D_\Lt )$.
Meanwhile, since $B'$ lies strictly in side $D'_\Lt B_\Lt$ of $\triangle CB_\Lt D'_\Lt$, we have $\phi_\Lt (D_\can )>0$.
Thus we have $\phi_\Lt (D_\Lt )<\phi_\Lt (D_\can )$ as desired.

Next suppose $\phi'_\Lt (B'_\Lt )<\phi'_\Lt (D'_\Rt)<\phi'_\Lt (D'_\Lt )$,
where $\phi'_\Lt (D'_\Rt)<\phi'_\Lt (B'_\Rt)$ also holds due to Proposition $\ref{prop:rel_B'_D'}$.
Then since $B'$ lies strictly in side $D'_\Rt B_\Rt$ of $\triangle CD'_\Rt B_\Rt$, we have $\phi_\Lt (D_\Rt )<\phi_\Lt (D_\can )$ as desired.

This completes the proof of Theorem $\ref{thm:exist_pos_can}$.
\end{proof}
\addtocounter{theorem}{1}
\begin{figure}[htbp]
\centering\includegraphics[width=0.75\hsize]{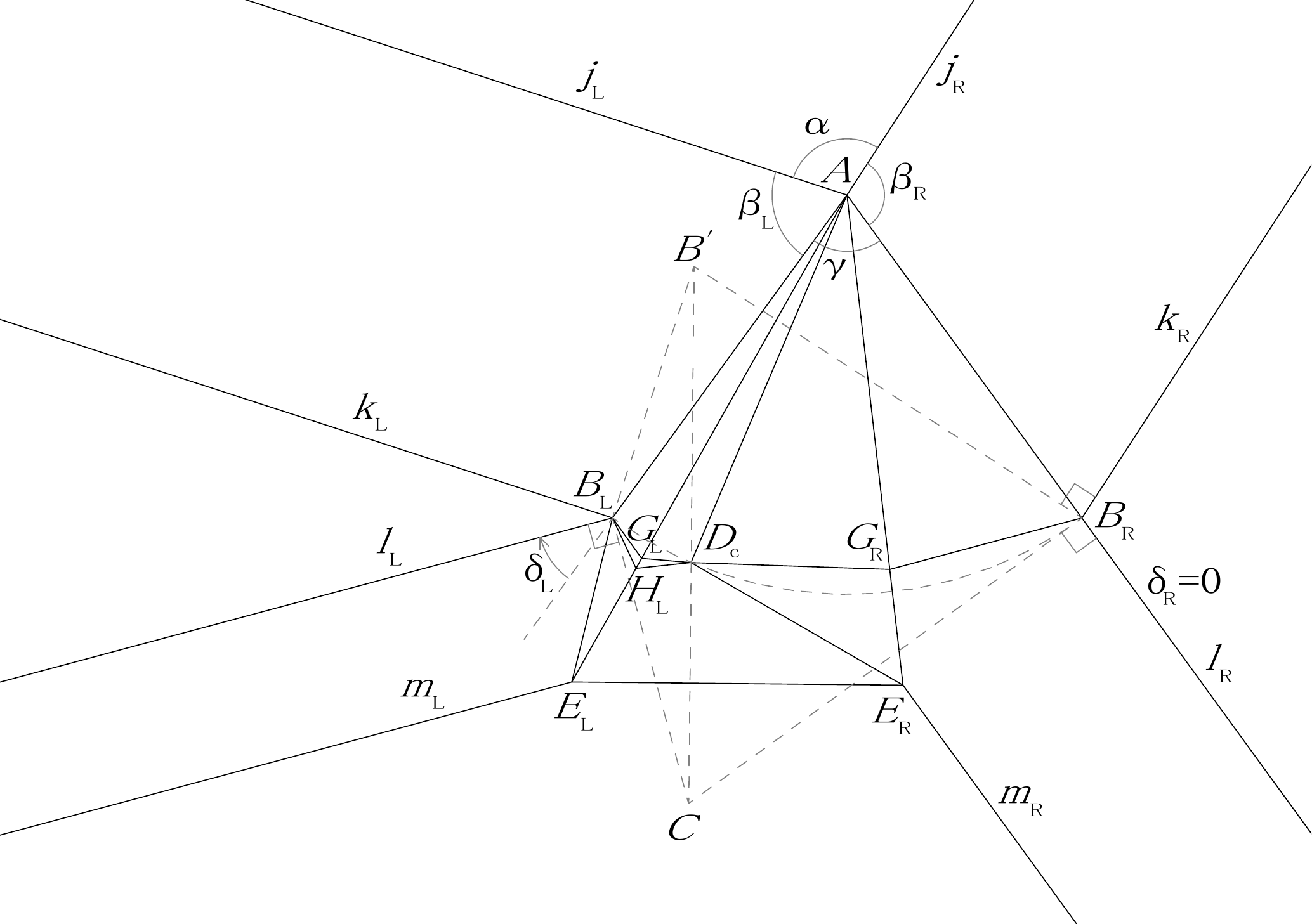}
\caption{Crease pattern of a canonical positive origon gadget}
\label{fig:pos_can_CP}
\end{figure}
We shall say that a positive origon gadget obtained in Construction $\ref{const:pos}$ by setting $D=D_\can$
is a \emph{canonical} positive origon (gadget), whose existence is ensured by Theorem $\ref{thm:exist_pos_can}$.
Also, we shall say that the canonical positive and negative origon resulting from a common development as in Figure $\ref{fig:dev_0}$
together forms a \emph{canonical pair}.
We show in Figure $\ref{fig:pos_can_CP}$ the canonical positive origon obtained from the development in Figure $\ref{fig:dev_0}$,
which forms a canonical pair together with the canonical negative one in Figure $\ref{fig:neg_can_CP}$.
Noting Convention $\ref{conv}$, $(3)$ on the orientation of the development of a $3$D gadget,
we see that the positive and the negative origon which forms a canonical pair engage with each other.

As an application of canonical pairs of origons, we can construct a positive and a negative extrusion made of origon gadgets
from a common crease pattern by replacing each origon with its canonical counterpart as long as there arise no interferences.
(For the details of the interferences caused by positive and negative origon gadgets, see \cite{Doi20}, Section $5$ and \cite{Doi21}, Section $3.4$.)
In particular, by combining Tables $\ref{tbl:pos_assign}$ and $\ref{tbl:neg_can_assign}$,
we obtain a `hybrid' crease pattern from which we can fold both a positive and a negative origon which together form a canonical pair.
We show in Table $\ref{tbl:hybrid}$ the hybrid creases for a canonical pair.
\addtocounter{theorem}{1}
\begin{table}[h]
\begin{tabular}{c|c|c}
&$\delta_\sigma =0$&$\delta_\sigma >0$\\ \hline
\multirow{2}{*}{common}&\multicolumn{2}{c}{$j_\sigma ,k_\sigma ,\ell_\sigma ,m_\sigma ,AB_\sigma ,AE_\sigma ,E_\Lt E_\Rt$}\\ \cline{2-3}
&---&$B_\sigma E_\sigma$\\ \hline
\multirow{2}{*}{positive}&\multicolumn{2}{c}{$AD_\can ,B_\sigma G_\sigma ,D_\can G_\sigma$}\\ \cline{2-3}
&$D_\can E_\sigma$&$B_\sigma H_\sigma ,D_\can H_\sigma$\\ \hline
\multirow{2}{*}{negative}&\multicolumn{2}{c}{$B_\sigma G'_\sigma ,B_\sigma P_\sigma ,G'_\Lt G'_\Rt ,P_\Lt P_\Rt$}\\ \cline{2-3}
&$B_\sigma E_\sigma$&---
\end{tabular}\vspace{0.5cm}
\caption{Hybrid creases for a canonical pair of a positive and a negative origon gadget}
\label{tbl:hybrid}
\end{table}
In general, $G_\sigma$ in the development of the positive origon and $G'$ in that of the negative origon are not identical,
nor are the creases through $G_\sigma$ and $G'_\sigma$.
However, in the following situation, we can reduce the creases in Table $\ref{tbl:hybrid}$.
\begin{proposition}
Consider a canonical pair of a positive and a negative origon gadget constructed from a common development as in Figure $\ref{fig:dev_0}$.
If $\beta_\Lt =\beta_\Rt =\pi /2$, then we have $G_\Lt =G'_\Lt$ and $G_\Rt =G'_\Rt$.
Also, $D_\can$ is the intersection point of segment $AC$ and minor arc $\arc{B_\Lt B_\Rt}$,
so that we have $\phi_\sigma (D_\can )=\gamma_\sigma$ and $\psi_\sigma (D_\can )=0$ for both $\sigma =\Lt ,\Rt$,
where $\gamma_\sigma$ is given by $\eqref{eq:gamma_sigma}$.
\end{proposition}
\begin{proof}
This follows easily from $B'=A$.
\end{proof}
\section{Conclusion}\label{sec:concl}
In this paper we obtained a geometric construction of canonical positive and negative origon gadgets.
For this purpose, we found points $B'$ and $D'_\Lt ,D'_\Rt$ with which we can determine 
the canonical dividing point $D_\can$ and the critical dividing points $D_\Lt ,D_\Rt$
by drawing a segment from the respective points to $C$ and taking the intersection with minor arc $B_\Lt B_\Rt$.
It is more or less surprising that point $D'_\sigma$ turned out to satisfy a simple equation $\phi'_\sigma (D'_\sigma )=2\beta_\sigma$
in the proof of Proposition $\ref{prop:tan_rho_D'}$ after some calculation.
Using this geometric interpretation of the dividing points, we also gave simplified proofs of results obtained in our previous papers.

Thus we found a canonical choice 
among an infinite number of possible choices of the dividing point for a positive origon gadget with a given set of angle parameters.
The canonical dividing point is different from the three practical choices of the dividing point suggested in \cite{Doi20}
which yield critical, orthogonal and balanced gadgets.
The notion of canonical pairs enables us to treat positive and negative origons on a one-to-one basis.

Although our results here do not extend the range of $3$D objects which we can extrude with $3$D gadgets,
we achieved a better geometric understanding of origon gadgets, which will also be useful in future developments of origami extrusions.

\end{document}